\documentclass[showpacs]{revtex4}
\usepackage{amsmath,amsthm}
\usepackage{setspace}
\usepackage{graphicx}
\usepackage{epstopdf}
\newcommand{\bra}[1]{\langle#1|}
\newcommand{\ket}[1]{|#1\rangle}
\usepackage{amsfonts}
\newtheorem{theorem}{Theorem}[section]
\newtheorem{lemma}[theorem]{Lemma}
\newtheorem{proposition}[theorem]{Proposition}

\theoremstyle{remark}
\newtheorem{remark}[theorem]{Remark}

\theoremstyle{definition}
\newtheorem{definition}[theorem]{Definition}

\theoremstyle{example}
\newtheorem{example}[theorem]{Example}

\theoremstyle{notation}

\begin{document}

\title{An analytic function approach to weak mutually unbiased bases}
\author{T. Olupitan, C. Lei, A. Vourdas\\
Department of Computing\\University of Bradford\\ Bradford BD7 1DP, UK}
\begin{abstract}
Quantum systems with variables in ${\mathbb Z}(d)$ are considered, and three different structures are studied.
The first is weak mutually unbiased bases,  for which the absolute value of the overlap of any two vectors in two different bases is $1/\sqrt{k}$ (where $k|d$) or $0$.
The second is maximal lines through the origin in the ${\mathbb Z}(d)\times {\mathbb Z}(d)$ phase space. 
The third is an analytic representation in the complex plane based on Theta functions, and their zeros. 
It is shown that there is a correspondence (triality) that links strongly these three apparently different structures.
For simplicity, the case where $d=p_1\times p_2$, where $p_1,p_2$ are odd prime numbers different from each other, is considered.

\end {abstract}
\pacs{03.65.Aa, 02.10.De}

\maketitle

\section{Introduction}

After the pioneering work by Schwinger \cite{SCH}, there has been a lot of work on various aspects of a quantum system
$\Sigma (d)$ with variables in ${\mathbb Z}(d)$ (the ring of integers modulo $d$), described
with a $d$-dimensional Hilbert space $H(d)$.
The work combines Quantum Physics with Discrete Mathematics and has applications to areas like quantum information, quantum cryptography, quantum coding, etc
(for reviews see \cite{1,2,3,4,5,6,7}).

A deep problem in this area is mutually unbiased bases \cite{m1,m2,m3,m4,m5,m6,m7,m8,m9,m10}. 
It is a set of bases, for which the absolute value of the overlap of any two vectors in two different bases is 
$1/\sqrt{d}$.
It is known that the number ${\mathfrak M}$ of mutually unbiased bases satisfies the inequality ${\mathfrak M} \le d+1$, and that
when $d$ is a prime number ${\mathfrak M}=d+1$.
What makes the case of prime $d$ special, is that ${\mathbb Z}(d)$ becomes a field, which is a stronger mathematical structure than a ring.
For the same reason, if we consider quantum systems with variables in the Galois field $GF(p^e)$ (where $p$ is a prime number),
the number of mutually unbiased bases is ${\mathfrak M} = p^e+1$.
The study of mutually unbiased bases for non-prime $d$, in which case ${\mathbb Z}(d)$ is a ring (but not a field), is a very difficult problem.
It is also related to the subjects of $t$-designs\cite{B,B1} and latin squares\cite {LS}.

Recent work \cite{W1,W2} introduced a weaker concept called weak mutually unbiased bases (WMUB).
It is a set of bases, for which the absolute value of the overlap of any two vectors in two different bases is 
$1/\sqrt{k}$, where $k|d$ ($k$ is a divisor of $d$), or zero. It has been shown that there are $\psi(d)$ (the Dedekind $\psi$-function) WMUBs.
This work has also studied the phase space ${\mathbb Z}(d)\times {\mathbb Z}(d)$ as a finite geometry ${\mathcal G}(d)$.

There exists much literature on finite geometries.
They consist of a finite number of points and lines which obey certain axioms (e.g., \cite{f1,f2,f3} in a mathematics context, and \cite{r1,r2,r3,r4} in a physics context).
Most of this work is on near-linear geometries, where two lines have at most one point in common.
The ${\mathbb Z}(d)\times {\mathbb Z}(d)$ geometry is based on rings and it does not obey this axiom.
Two lines have in common a `subline' which consists of $k$ points, where $k|d$.
Refs\cite{W1,W2} have shown that there is a duality between WMUBs in $H(d)$ and lines in ${\mathcal G}(d)$.
This shows a deep connection between finite quantum systems and the geometries of their phase spaces.

A very different problem is the use of analytic functions in the context of physical systems. 
After the pioneering work by Bargmann\cite{A0,A00} for the harmonic oscillator, analytic representations have been used with various quantum systems
(e.g., \cite{A1,A2,A3,A4,A5,A6,A7,A8,A9,A10}). 
In particular the zeros of the analytic functions have been used for the derivation of physical results.
For example, there are links between the growth of analytic functions at infinity, and the density of their zeros\cite{C1,C2,C3}, which lead to 
criteria for the overcompleteness or undercompleteness of a von Neumann lattice of coherent states.

Refs\cite{AN4,AN5} have studied analytic representations for quantum systems with variables in ${\mathbb Z}(d)$, using Theta functions \cite{THETA}
(see also ref\cite{AN6}).
Quantum states are represented with analytic functions in the cell ${\mathfrak S}=[0,d)\times [0,d)$ in the complex plane (i.e., in a torus).
These analytic functions have exactly $d$ zeros in the cell ${\mathfrak S}$, which determine uniquely the state of the system.

In this paper we use this language of analytic functions for the study WMUBs.
We show that: 
\begin{itemize}
\item
Each of the $d$ vectors in a WMUB has $d$ zeros on a straight line.
\item
In a given WMUB, the various vectors have zeros on parallel lines.
In different WMUBs, the slope of the lines of zeros, is different.
\item
The $d^2$ zeros in each WMUB, form a regular lattice in the cell ${\mathfrak S}$, which is the same for all WMUBs.
\end{itemize}

Based on these results we show that there is a triality between 
\begin{itemize}
\item
WMUBs
\item
Lines through the origin in the finite geometry ${\mathcal G}(d)$ of the phase space
\item
Sets of parallel lines of zeros of the vectors in WMUBs in the cell ${\mathfrak S}$
\end{itemize}
These three mathematical objects, which are very different from each other, have the same mathematical structure.
The work links the theory of analytic functions and their zeros, to finite quantum systems, finite geometries and more generally to Discrete Mathematics.

In order to avoid a complicated notation, in all sections except section II, 
we consider the case that $d=p_1\times p_2$, where $p_1,p_2$ are odd prime numbers, different from each other
(in section II we state in each subsection, what values $d$ takes).
All results are generalizable to the case $d=p_1\times...\times p_N$, where $d$ is an odd integer (see discussion).
In the case of even dimension $d$ (e.g., \cite{zak}) , some aspects of the formalism of finite quantum systems require special consideration, and 
further work is needed in order to extend the ideas of the present paper, to this case. 
Also when $d$ contains powers of prime numbers, further work is needed (based on labeling with elements of Galois fields).

In section 2 we introduce very briefly finite quantum systems, their analytic representation, and mutually 
unbiased bases,  in order to define the notation. In section 3 we review briefly the formalism of weak mutually unbiased bases.
An important ingredient is the factorization of $\Sigma (d)$ in terms of smaller systems $\Sigma (p_1)$ and $\Sigma (p_2)$, which is based on the 
Chinese remainder theorem, and its use by Good \cite{Good} in the context of finite Fourier transforms.
In section 4, we use the analytic representation to study WMUBs, and prove the results that we mentioned above.
We conclude in section 5, with a discussion of our results.

\section{Preliminaries}

\subsection{Analytic representation of quantum systems with variables in ${\mathbb Z}(d)$, with odd $d$}
We consider a finite quantum system with variables in ${\mathbb Z}(d)$ (the integers modulo $d$)\cite{1,2,3,4,5,6,7}.
Let $\ket{X; n}$ the basis of position states in the $d$-dimensional Hilbert space $H(d)$, and $\ket{P; n}$ the basis of momentum states:
\begin{eqnarray}\label{Fou}
\ket{P;n}={\cal F}|{X};n\rangle ;\;\;\;\;\;\;{\cal F}=d^{-1/2}\sum _{m,n}\omega (mn)\ket{X;m}\bra{X;n};\;\;\;\;
\omega (n)=\exp \left [i \frac{2\pi n}{d}\right ]
\end{eqnarray}
Here ${\cal F}$ is the finite Fourier transform.
Displacement operators are given by
\begin{eqnarray}\label{99}
{D}(\alpha,\beta)={Z}^\alpha {X}^\beta \omega (-2^{-1}\alpha \beta );\;\;\;\;\;\alpha, \beta \in {\mathbb Z}(d)
\end{eqnarray}
where
\begin{eqnarray}\label{69}
&&{Z}=\sum _{n}\omega (n)|{X};n\rangle \langle {X};n|=\sum _n\ket{{P};n+1}\bra{{P};n}\nonumber\\
&&{X}=\sum _{n}\omega (-n)|{P};n\rangle \langle {P};n|=\sum _n \ket{{X};n+1}\bra{{ X};n}\nonumber\\
&&{X}^{d}={Z}^{d}={\bf 1};\;\;\;\;\;\;
{X}{Z} = {Z}{X}\omega (-1)
\end{eqnarray}
The $\{{D}(\alpha,\beta)\omega (\gamma)\}$ form a representation of the Heisenberg-Weyl group in this context.

Let $\ket {g}$ be an arbitrary state
\begin{eqnarray}\label{6}
&&\ket {g}=\sum _m g_m\ket {X;m}=\sum _m{\widetilde g}_m\ket {P;m};\;\;\;\;\;\;\sum _m |g_m|^2=1\nonumber\\
&&{\widetilde g}_m=d^{-1/2}\sum _n\omega (-mn)g_n
\end{eqnarray}
We use the notation (star indicates complex conjugation)
\begin{eqnarray}
\ket {g^*}=\sum _m g_m^*\ket {X;m};\;\;\;\;\;\;
\bra {g}=\sum _m g_m^*\bra {X;m};\;\;\;\;\;\;\bra {g^*}=\sum _m g_m\bra {X;m}
\end{eqnarray}
We represent the state $\ket{g}$ with the function
\begin{eqnarray}\label{aaa1}
G(z)=\pi^{-1/4} \sum_{m=0}^{d-1} g_m^*\;\Theta_3 \left [\frac{\pi m}{d}-z\frac{\pi}{d};\frac{i}{d}\right ]
\end{eqnarray}
where $\Theta_3$ is Theta function \cite{THETA}:
\begin{eqnarray}\label{theta}
&&\Theta _3(u,\tau)=\sum_{n=-\infty}^{\infty}\exp(i\pi \tau n^2+i2nu)
\end{eqnarray}
Theta functions are `Gaussian functions wrapped on a circle', and in our case on a `discretized circle'.
Their periodicity properties are:
\begin{eqnarray}
&&\Theta _3(u+\pi,\tau)=\Theta _3(u,\tau +2)=\Theta _3(u,\tau)\nonumber\\
&&\Theta _3(u+\tau \pi,\tau)=\Theta _3(u,\tau)\exp[-i(\pi \tau+2u)]
\end{eqnarray}
For later use we mention that
\begin{equation}\label{y}
\Theta_3(u, \tau) = \left(-i\tau\right)^{-1/2}\exp\left[ \frac{u^2}{i\pi\tau}\right]\Theta_3\left( \frac{u}{ \tau}, \frac{-1}{ \tau}\right),
\end{equation}
and that their zeros are
\begin{eqnarray}\label{eq1}
\zeta _{MN}=(2M-1)\frac{\pi}{2}+(2N-1)\frac{i\pi}{2d}.
\end{eqnarray}
$G(z)$ is an analytic function and obeys the periodicity relations
\begin{eqnarray} \label{periodicity}
 &&G( z+d ) = G(z)\nonumber\\
 &&G( z+ i d)  = G(z)\exp\left( -\pi d -2i\pi z\right).
\end{eqnarray}
The scalar product is given by
\begin{eqnarray}\label{scalar}
 \langle g_2| g_1^* \rangle &=& \frac{ \sqrt{2\pi}}{d^{5/2}} \int_{\mathfrak S} dz_Rdz_I\exp\left( \frac{-2\pi}{d}z_I^2\right) G_1(z)G_2(z^\ast)=\sum _{m\in {\mathbb Z}(d)}g_{2m}^*g_{1m}^*
\end{eqnarray}
where $z_R, z_I$ are the real and imaginary parts of $z$.
${\mathfrak S}_{MN}=[M d,(M+1)d)\times [N d,(N+1)d)$ is a cell in the complex plane and
$(M,N)$ are integers labelling the cell. In the case $M=N=0$ we use the simpler notation ${\mathfrak S}$.
The proof of Eq.(\ref{scalar}) is based on the orthogonality of Theta functions.

The analytic function $G(z)$ has exactly $d$ zeros $\zeta _r$ in each cell and the sum of these zeros is \cite{AN4,AN5,AN6}
\begin{eqnarray}\label{con}
\sum _{r =1}^d \zeta _r= d(M+iN)+\frac{d^2}{2}(1+i).
\end{eqnarray}
So in each cell $d-1$ zeros are independent, and the last is determined by this constraint.

\subsection{Mutually unbiased bases using $Sp(2,{\mathbb Z}(d))$ symplectic transformations, with odd prime $d$}\label{S200}

In this subsection $d$ is a prime number and therefore ${\mathbb Z}(d)$ is a field.
Symplectic transformations are defined as
\begin{eqnarray}\label{nnn}
&&X'=S(\kappa, \lambda|\mu ,\nu)\;X\;[S(\kappa, \lambda|\mu ,\nu)]^{\dagger}=D(\lambda, \kappa)\nonumber\\
&&Z'=S(\kappa, \lambda|\mu ,\nu)\;Z\;[S(\kappa, \lambda|\mu ,\nu)]^{\dagger}=D(\nu,\mu)\nonumber\\
&&\kappa \nu-\lambda \mu=1;\;\;\;\;\;\kappa, \lambda, \mu, \nu \in {\mathbb Z}(d)
\end{eqnarray}
They form a representation of the $Sp(2,{\mathbb Z}(d))$ group. 
Eqs.(\ref{nnn}) define uniquely (up to a phase factor) the symplectic transformations. $S(\kappa, \lambda|\mu ,\nu)$ is given by\cite{2}
\begin{eqnarray}
&&S(\kappa, \lambda|\mu ,\nu)=S(1,0|\xi _1,1)S(1,\xi_2|0,1)S(\xi_3,0|0,\xi _3^{-1})\nonumber\\
&&S(1,0|\xi _1,1)=\sum _n\ket{X;n}\bra{X;\xi _1n}\nonumber\\
&&S(1,\xi_2|0,1)=\sum _n\omega (2^{-1}\xi_2 n^2)\ket{X;n}\bra{X;n}\nonumber\\
&&S(\xi_3,0|0,\xi _3^{-1})=\sum _n\omega (2^{-1}\xi_3 n^2)\ket{P;n}\bra{P;n}
\end{eqnarray}
where
\begin{eqnarray}
\xi _1=\kappa \mu (1+\lambda \mu)^{-1};\;\;\;\;
\xi _2=\lambda \kappa ^{-1} (1+\lambda \mu);\;\;\;\;
\xi _3=\kappa (1+\lambda \mu)^{-1}.
\end{eqnarray}
The multiplication rule is given by
\begin{eqnarray}\label{500}
&&S(\kappa _1, \lambda _1|\mu _1,\nu _1)S(\kappa _2, \lambda _2|\mu _2,\nu _2)=S(\kappa, \lambda|\mu ,\nu)\nonumber\\
&&\left (
\begin{array}{cc}
\kappa& \lambda\\
\mu&\nu
\end{array}
\right )=
\left (
\begin{array}{cc}
\kappa _2& \lambda _2\\
\mu _2&\nu _2
\end{array}
\right )
\left (
\begin{array}{cc}
\kappa _1& \lambda _1\\
\mu _1&\nu _1
\end{array}
\right )
\end{eqnarray}

We consider the following special case of symplectic transformations:
\begin{eqnarray}\label{DD175}
&&X'=S(0,-\mu ^{-1}|\mu,\nu)\;X\;[S(0, -\mu ^{-1}|\mu;\nu)]^{\dagger}
=Z^{-\mu ^{-1}};\;\;\;\;\mu, \nu\in {\mathbb Z}(d)\nonumber\\
&&Z'=S(0,-\mu ^{-1}|\mu,\nu)\;Z\;[S(0,-\mu ^{-1}|\mu,\nu)]^{\dagger}=X^{\mu}Z^{\nu}\omega (2^{-1}\mu \nu)
\end{eqnarray}
We note that $S(0,-1|1,0)={\cal F}^{-1}$.
We can show that these transformations preserve Eq.(\ref{69}).
Acting with them on the position basis, we get new bases:
\begin{eqnarray}\label{Sym7}
\ket{X(\mu,\nu);m} \equiv S(0,-\mu ^{-1}|\mu,\nu)\ket{X;m};\;\;\;\;\;\nu=0,..., d-1
\end{eqnarray}
We note that $\ket{X(\mu,0);m}=\ket{P; -\mu ^{-1}m}$.

\begin{lemma}\label{l}
\begin{eqnarray}\label{3cv}
\ket{X(\mu,\nu);m}=\frac{1}{ \sqrt{d}} \sum_{j=0}^{d-1} \omega[\mu ^{-1}\phi (m,j,\nu)]|X;j\rangle ;\;\;\;\;\phi (m,j,\nu)=-jm+2^{-1}\nu j^2
\end{eqnarray}
\end{lemma}
\begin{proof}
We first prove that these states are eigenstates of $Z'=X^{\mu}Z^{\nu}\omega (2^{-1}\nu \mu)$.
\begin{align}
&&Z'\ket{X(\mu,\nu);m}= \frac{1}{ \sqrt{d}} \omega (2^{-1}\nu \mu)\sum_{j=0}^{d-1} \omega [\mu ^{-1}\phi (m,j,\nu)]X^{\mu}Z^{ \nu}|X;j\rangle\nonumber\\&&
= \frac{1}{ \sqrt{d}}\omega (2^{-1}\nu \mu) \sum_{j=0}^{d-1} \omega [\mu ^{-1}\phi (m,j,\nu)]\omega( \nu j)|X;j+\mu\rangle
\end{align}
We now change variables $j' = j+\mu$ and we get
\begin{align}
&Z'\ket{X(\mu,\nu);m}=\omega(m)\ket{X(\mu,\nu);m}
\end{align}
We next show that $X'\ket{X(\mu,\nu);m}=\ket{X(\mu,\nu);m+1}$.
\begin{align}
&Z^{-\mu ^{-1}}\ket{X(\mu,\nu);m}= \frac{1}{ \sqrt{d}} \sum_{j=0}^{d-1} \omega [\mu ^{-1}\phi (m,j,\nu)]Z^{-\mu ^{-1}}|X;j\rangle = \frac{1}{ \sqrt{d}} \sum_{j=0}^{d-1} \omega [\mu ^{-1}\phi (m,j,\nu)]\omega(-j\mu ^{-1})|X;j\rangle \nonumber \\
&= \frac{1}{ \sqrt{d}} \sum_{j=0}^{d-1} \omega [\mu ^{-1}\phi (m+1,j,\nu)]|X;j\rangle =\ket{X(\mu, \nu);m+1}
\end{align}
\end{proof}
It is known that for a prime number $d$ there are $d+1$ mutually unbiased bases given by
\begin{eqnarray}\label{S100}
B(\mu,-1)=\{\ket{X;m}\};\;\;\;\;B(\mu, \nu)=\{\ket{X(\mu, \nu);m}\};\;\;\;\;\;\nu=0,1,...,d-1.
\end{eqnarray}
Here $\mu $ is fixed. $B(\mu, 0)$ is the basis of momentum states $\{\ket{X(\mu, 0);m}=\ket{P;-\mu ^{-1} m}\}$.
They are mutually unbiased bases\cite{m1,m2,m3,m4,m5,m6,m7,m8,m9}, because for all $\nu \ne \nu'$ and for all $n,m$
\begin{eqnarray}\label{8}
|\langle X(\mu, \nu);n\ket{X(\mu, \nu ');m}|=d^{-1/2}
\end{eqnarray}

\subsection{Maximal lines through the origin in ${\cal G}(d)$}

Various aspects of the ${\mathbb Z}(d)\times {\mathbb Z}(d)$ phase space as a finite geometry ${\cal G}(d)$ have been studied in\cite{W1,W2}.
A special class of finite geometries which has been studied extensively in the discrete mathematics literature\cite{f1,f2,f3} is the near-linear geometries, which have the axiom
that two lines have at most one point in common. These geometries are intimately related to fields.
The ${\cal G}(d)$ geometry does not obey this axiom, is based on rings and it is a non-near-linear geometry.
Two lines through the origin have a `subline' in common, which consists of $k$ points, where $k|d$. If $d$ is a prime number, $k$ is $1$ (in which case the lines have one point in common) or $d$ (in which case the lines are identical),
and this is the near-linear geometry.

In this subsection $d=p_1\times p_2$, where $p_1,p_2$ are odd prime numbers different from each other.
The ${\cal G}(d)$ is defined as $(P(d),L(d))$ where $P(d)$ is the set of the $d^2$ points $(m,n)\in {\mathbb Z}(d)\times {\mathbb Z}(d)$
and $L(d)$ is the set of lines.
A maximal line through the origin is the set of $d$ points
\begin{eqnarray}
&&L(\rho , \sigma )=\{(r\rho ,r\sigma )\;|\;r \in {\mathbb Z}(d)\};\;\;\;\;\rho , \sigma \in {\mathbb Z}(d).
\end{eqnarray}
If $\tau$ is an invertible element of ${\mathbb Z}(d)$ then $L(\rho,\sigma )$ is the same line as $L(\tau \rho,\tau \sigma )$.
An example of a non-maximal line is $L(p_1,\tau p_1)$ (it has only $p_2$ points).
There are $\psi(d)=(p_1+1)(p_2+1)$ (the Dedekind psi function) maximal lines through the origin in ${\mathcal G}(d)$. 

Symplectic transformations on a point $(\rho , \sigma ) \in {\mathbb Z}(d)\times {\mathbb Z}(d)$ are given by
\begin{eqnarray}
&&{\cal S}(\kappa, \lambda|\mu ,\nu)(\rho , \sigma )=
(\rho , \sigma )\left (
\begin{array}{cc}
\kappa& \lambda\\
\mu&\nu
\end{array}
\right )
=(\kappa \rho +\mu \sigma, \lambda \rho +\nu \sigma )\nonumber\\&&\kappa \nu-\lambda \mu=1;\;\;\;\;\;\kappa, \lambda, \mu, \nu \in {\mathbb Z}(d)
\end{eqnarray}
where we represent points with rows and act on the right, or by
\begin{eqnarray}
&&{\cal S}(\kappa, \lambda|\mu ,\nu)
\left (
\begin{array}{c}
\rho\\
\sigma
\end{array}
\right )=
\left (
\begin{array}{cc}
\kappa& \lambda\\
\mu&\nu
\end{array}
\right )^T
\left (
\begin{array}{c}
\rho\\
\sigma
\end{array}
\right )
=
\left (
\begin{array}{c}
\kappa \rho +\mu \sigma\\
\lambda \rho +\nu \sigma
\end{array}
\right )
\end{eqnarray}
where we represent points with columns and act with the transposed matrix on the left.
With this notation we get the same multiplication rule as in Eq.(\ref{500}).
We have here a representation of the $Sp(2,{\mathbb Z}(d))$ group.

Symplectic transformations on points lead to symplectic transformations on lines:
\begin{eqnarray}\label{efb}
{\cal S}(\kappa, \lambda|\mu ,\nu)L(\rho , \sigma )=L(\kappa \rho +\mu \sigma, \lambda \rho +\nu \sigma ).
\end{eqnarray}

\section{Factorization}
In the rest of the paper $d=p_1\times p_2$, where $p_1,p_2$ are odd prime numbers different from each other.
\subsection{Factorization of the system in terms of smaller systems}

Based on the Chinese remainder theorem, and following ref.\cite{Good} on the factorization of finite Fourier transforms, we introduce two bijective maps between
${\mathbb Z}(d)$ and ${\mathbb Z}(p_1)\times {\mathbb Z}(p_2)$:
\begin{eqnarray}\label{map1}
&&m\leftrightarrow (m_1,m_2);\;\;\;\;\;m_i=m ({\rm mod}\ p_i);\;\;\;\; m=m_1s_1+m_2s_2\;({\rm mod}\ d)\nonumber\\
&&m\in {\mathbb Z}(d);\;\;\;\;\;m_i\in {\mathbb Z}(p_i),
\end{eqnarray}
and
\begin{eqnarray}\label{map2}
&&m\leftrightarrow (\overline m_1,\overline m_2);\;\;\;\;\;\overline m_i=mt_i=m_it_i({\rm mod}\ p_i);\;\;\;\;m=\overline m_1 r_1+\overline m_2r_2\;({\rm mod}\ d)\nonumber\\
&&m\in {\mathbb Z}(d);\;\;\;\;\;{\overline m}_i\in {\mathbb Z}(p_i).
\end{eqnarray}
Here $r_i,t_i,s_i$ are the constants
\begin{equation}\label{20}
r_1=\frac{d}{p_1}=p_2;\;\;\;\;r_2=\frac{d}{p_2}=p_1;\;\;\;\;t_i r_i=1\;( mod\ p_i);\;\;\;\;\;s_i=t_i r_i\in {\mathbb Z}(d).
\end{equation}
We note that
\begin{eqnarray}\label{20A}
&&s_1s_2=0\;({\rm mod}\ d);\;\;\;\;\;s_1^2=s_1\;({\rm mod}\ d)\;\;\;\;\;s_2^2=s_2\;({\rm mod}\ d)\;\;\;\;\;s_1+s_2=1\;({\rm mod}\ d)\nonumber\\
&&p_2s_1=p_2\;({\rm mod}\ d);\;\;\;\;p_1s_2=p_1\;({\rm mod}\ d);\;\;\;\;p_1s_1=p_2s_2=0\;({\rm mod
}\ d).
\end{eqnarray}
Also for the map of Eq.(\ref{map1})
\begin{eqnarray}\label{20B}
m+\ell \;\leftrightarrow \;({m_1+\ell _1},{m_2+\ell _2});\;\;\;\;\;m\ell \;\leftrightarrow \;({m_1\ell _1},{m_2\ell _2}),
\end{eqnarray}
and for the map of Eq.(\ref{map2})
\begin{eqnarray}\label{20C}
 m+\ell \;\leftrightarrow \;(\overline {m}_1+\overline {\ell }_1,\overline {m}_2+\overline{\ell} _2)
;\;\;\;\;\;
 m\ell \;\leftrightarrow \;(\overline {m}_1\ell _1, \overline {m}_2\ell _2)
\end{eqnarray}
Using the notation $\omega _i(n)=\exp (\frac{2\pi n_i}{p_i})$ where $n_i\in {\mathbb Z}(p_i)$, we can show that
\begin{eqnarray}\label{20D}
\omega (mn)=\omega _1(m_1{\overline n}_1)\omega _2(m_2{\overline n}_2).
\end{eqnarray}
Eqs.(\ref{20A}), (\ref{20B}), (\ref{20C}), (\ref{20D}), are important for the proof of various relations below.

We introduce a bijective map from $H(d)$ to $H(p_1)\otimes H(p_2)$ as follows\cite{2}.
We use the map of Eq.(\ref{map2}) for position states:
\begin{equation}
\ket{X;m}\;\;\leftrightarrow\;\;\ket{X_1;{\overline m}_1} \otimes \ket{X_2;{\overline m}_2},
\end{equation}
where $\ket{X_i;{\overline m}_i}$ are position states in $H(p_i)$.
Using Eq.(\ref{20D}) we prove that the corresponding map for momentum states, is based on the map of Eq.(\ref{map1}), and it is given by
\begin{equation}
\ket{P;m}\;\;\leftrightarrow\;\;\ket{P_1;{m}_1}\otimes \ket{P_2;{m}_2}
\end{equation}
where $\ket{P_i;{m}_i}$ are momentum states in $H(p_i)$.
The Fourier transform between position and momentum states, implies that if the map of Eq.(\ref{map2}) is used for position states, then the map of Eq.(\ref{map1})
should be used for momentum states. 

For later use we also factorize the symplectic transformations.
The $Sp(2,{\mathbb Z}(d))$ is factorized as
$Sp(2,{\mathbb Z}(p_1))\times Sp(2,{\mathbb Z}(p_2))$, as follows (proposition 3.1 in \cite{SV2}):
\begin{equation}\label{345}
S(\kappa, \lambda |\mu , \nu)=S(\kappa _1, \lambda _1r_1|\overline \mu _1, \nu _1)\otimes S(\kappa _2, \lambda _2r_2|\overline \mu _2, \nu _2)
\end{equation}
where the $\kappa _i$, $\lambda _i$, ${\overline \mu}_i$ , $\nu_i$ are related to $\kappa$, $\lambda$, $\mu $, $\nu$, as in Eqs(\ref{map1}),(\ref{map2}).
Below we need the special cases
\begin{eqnarray}\label{135}
&&S(0,-\mu ^{-1}|\mu,\nu)=S(0,-1|1, \nu _1)\otimes  S(0,-1|1, \nu _2)\nonumber\\
&&\nu=\nu_1s_1+\nu_2s_2;\;\;\;\;\;\mu=p_1+p_2;\;\;\;\;\;\mu ^{-1}=p_2^{-1}s_1+p_1^{-1}s_2\;({\rm mod} \;d),
\end{eqnarray}
and
\begin{eqnarray}\label{dfg}
&&S(\kappa,\lambda|\mu,\nu)={\bf 1}\otimes  S(0,-1|1, \nu _2)\nonumber\\
&&\kappa= s_1;\;\;\;\;\lambda=-s_2p_1^{-1}\;\;\;\;\;\mu=p_1;\;\;\;\;\;\;\nu=s_1+\nu_2s_2
\end{eqnarray}
and
\begin{eqnarray}\label{dfg1}
&&S(\kappa,\lambda|\mu,\nu)=S(0,-1|1, \nu _1)\otimes {\bf 1}\nonumber\\
&&\kappa= s_2;\;\;\;\;\lambda=-s_1p_2^{-1}\;\;\;\;\;\mu=p_2;\;\;\;\;\;\;\nu=s_2+\nu_1s_1.
\end{eqnarray}

As an example we consider the case that $d=21$, i.e., $p_1=3$ and $p_2=7$. Then
\begin{eqnarray}
&&r_1=7;\;\;\;t_1=1;\;\;\;\;s_1=7\nonumber\\
&&r_2=3;\;\;\;t_2=5;\;\;\;\;s_2=15\nonumber\\
&&\mu=10;\;\;\;\;-\mu ^{-1}=2
\end{eqnarray}
and we get
\begin{eqnarray}
&&S(0,2|10 , 7\nu _1+15 \nu_2)=S(0, -1|1, \nu _1)\otimes S(0, -1|1, \nu _2)\nonumber\\
&&S(7,9|3,7+15\nu _2)={\bf 1}\otimes  S(0,-1|1, \nu _2)\nonumber\\
&&S(15,-7|7, 15+7\nu _1)=S(0,-1|1,\nu_1)\otimes {\bf 1}
\end{eqnarray}

\subsection{Weak mutually unbiased bases}

For $d=p_1\times p_2$, references \cite{W1,W2} introduced in $H(d)=H(p_1)\otimes H(p_2)$ a weaker than mutually unbiased bases concept, 
called weak mutually unbiased bases (WMUB). 
They are tensor products of mutually unbiased bases in $H(p_i)$. They are given by
\begin{eqnarray}
&&\ket{{\cal X}(\nu _1,\nu _2);\overline m_1,\overline m_2}=\ket{X_1(\nu _1);\overline m_1}
\otimes \ket{X_2(\nu _2);\overline m_2}\nonumber\\
&&\ket{X_i(\nu _i);\overline m_i}=S(0,-1|1, \nu _i)\ket{X_i;\overline m_i};\;\;\;\;\;\overline m_i\in {\mathbb Z}(p_i).
\end{eqnarray}
We also include the $\nu _i=-1$, in which case $\ket{X_i(-1);\overline m_i}=\ket{X_i;\overline m_i}$.
Therefore $\nu _i=-1,...,p_i-1$.

In the special case $\nu _1=\nu _2=-1$ we get
\begin{eqnarray}
\ket{{\cal X}(-1,-1);\overline m_1,\overline m_2}=\ket{X_1(-1);\overline m_1}
\otimes \ket{X_2(-1);\overline m_2}=\ket{X_1;\overline m_1}
\otimes \ket{X_2;\overline m_2}
\end{eqnarray}
In the special case $\nu_1=\nu _2=0$ we get
\begin{eqnarray}
\ket{{\cal X}(0,0);\overline m_1,\overline m_2}=\ket{X_{1}(0);\overline m_1}
\otimes \ket{X_{2}(0);\overline m_2}=\ket{P_{1}; m_1}
\otimes \ket{P_{2}; m_2}
\end{eqnarray}

The overlap of two vectors in two different bases, is $0$ or $1/k$ where $k$ is a divisor of $d$:
\begin{eqnarray}\label{8}
|\langle { {\cal X}}(\nu _1,\nu _2);\overline m_1,\overline m_2\ket{{\cal X}(\nu _1',\nu _2');\overline r_1,\overline r_2}|^2=\frac{1}{k}\;\;{\rm or}\;\;0;\;\;\;\;k|d.
\end{eqnarray}
The strict requirement that the square of the absolute value of the overlap is $1/d$ in mutually unbiased bases, is replaced with the weaker requirement that it is
$1/k$ or $0$. And that is why we call them weak mutually unbiased bases. There are
$\psi (d)=(p_1+1)(p_2+1)$
weak mutually unbiased bases.

Taking into account Eq.(\ref{135}), we can relabel the $\ket{{\cal X}(\nu _1,\nu _2);\overline m_1,\overline m_2}$ as follows:
\begin{itemize}
\item
\begin{eqnarray}\label{np1}
&&\ket{{\cal X}(\nu _1,\nu _2);\overline m_1,\overline m_2}=S(0,-1|1,\nu _1)\ket{X;{\overline m}_1}\otimes S(0,-1|1,\nu _2)\ket{X;{\overline m}_2}\nonumber\\
&&=S(0,-\mu ^{-1}|\mu,\nu )\ket{X;m}=S(0,-1|1,\mu ^{-1}\nu )S(\mu ^{-1},0|0,\mu)\ket{X;m}\nonumber\\
&&=S(0,-1|1,\mu ^{-1}\nu )\ket{X;m\mu ^{-1}}\equiv \ket{{X}(1,\mu ^{-1}\nu );m \mu ^{-1}}\nonumber\\
&&\nu=\nu _1s_1+\nu _2s_2;\;\;\;\;\;\mu=p_1+p_2;\;\;\;\;\;\nu_i=0,..., p_i-1
\end{eqnarray}
Here we have used Eq.(\ref{135}), and $m$ is related to $\overline m_1,\overline m_2$ through Eq.(\ref{map2}). 
\item
\begin{eqnarray}\label{np2}
&&\ket{{\cal X}(-1,\nu _2);\overline m_1,\overline m_2}=\ket{X;{\overline m}_1}\otimes S(0,-1|1,\nu _2)\ket{X;{\overline m}_2}\nonumber\\&&=
\ket{X;{\overline m}_1}\otimes S(0,-1|1,\nu _2)\ket{X;{\overline m}_2}=S(\kappa, \lambda|\mu,\nu)\ket{X;m}=
\ket{{X}(p_1, s_1+\nu _2s_2 );m}\nonumber\\
&&\kappa=s_1;\;\;\;\;
\lambda=-s_2p_1^{-1};\;\;\;\;
\mu=p_1;\;\;\;\;
\nu=s_1+\nu_2s_2
\end{eqnarray}
Here we used Eq.(\ref{dfg}).
In a similar way we get
\begin{eqnarray}\label{np3}
\ket{{\cal X}(\nu _1,-1);\overline m_1,\overline m_2}=\ket{{X}(p_2, s_2+\nu _1s_1 );m}
\end{eqnarray}
\item
\begin{eqnarray}\label{np4}
\ket{{\cal X}(-1,-1);\overline m_1,\overline m_2}=\ket{X;{\overline m}_1}\otimes \ket{X;{\overline m}_2}=\ket{{X}(0,1);m}
\end{eqnarray}
\end{itemize}
There are $p_1p_2$ states in Eq.(\ref{np1}) (which have already been introduced in Eq.(\ref{S100})),
$p_1+p_2$ states in Eqs.(\ref{np2}),(\ref{np3}) and one state in Eq.(\ref{np4}).
Together they make the set of $\psi(d)=(p_1+1)(p_2+1)$ weak mutually unbiased bases.
\begin{remark}\label{zxd}
From the above it is clear that we use two different notations, `the factorized notation' (for which we use calligraphic letters) and the `unfactorized notation'.
In the unfactorized notation we have four different cases where different symplectic transformations act on the position states:
\begin{eqnarray}
&&\ket{{X}(1,\mu ^{-1}\nu);m\mu^{-1}}=S(0,-1|1,\mu ^{-1}\nu )\ket{X;m\mu ^{-1}}\nonumber\\
&&\ket{{X}(p_1, s_1+\nu _2s_2 );m}=S(s_1, -s_2p_1^{-1}|p_1,s_1+\nu _2s_2 )\ket{X;m}\nonumber\\
&&\ket{{X}(p_2, s_2+\nu _1s_1 );m}=S(s_2, -s_1p_2^{-1}|p_2,s_2+\nu _1s_1 )\ket{X;m}\nonumber\\
&&\ket{{X}(0,1);m}=\ket{X;{\overline m}_1}\otimes \ket{X;{\overline m}_2}
\end{eqnarray}
\end{remark}
In the `factorized notation' ${\cal B}(\nu _1, \nu _2)$ is the basis
\begin{eqnarray}
{\cal B}(\nu _1, \nu _2)=\{\ket{{\cal X}(\nu _1,\nu _2);\overline m_1,\overline m_2}\};\;\;\;\;\;\nu _i=-1,...,p_i-1.
\end{eqnarray}
In the `unfactorized notation'
\begin{eqnarray}
{B}(\mu, \nu )=\{\ket{{X}(\mu, \nu );m}\},
\end{eqnarray}
where $\mu$ takes the values $1, p_1, p_2, 0$.

The overlap of Eq.(\ref{8}) for vectors in two bases ${\cal B}(\nu _1, \nu _2)$ and ${\cal B}(\nu _1', \nu _2')$
takes one of the two values $\frac{r(\nu _1, \nu _2|\nu _1', \nu _2')}{d}$ or $0$.
We express this as
\begin{eqnarray}\label{jjj}
({\cal B}(\nu _1, \nu _2), {\cal B}(\nu _1', \nu _2'))=\frac{r(\nu _1, \nu _2|\nu _1', \nu _2')}{d}\;\;{\rm or}\;\;0
\end{eqnarray}
where
\begin{eqnarray}\label{200}
&&r(\nu _1, \nu _2|\nu _1', \nu _2')=1\;\;{\rm if}\;\;\nu _1\ne \nu_1' \;\;{\rm and}\;\;\nu _2\ne \nu_2'\nonumber\\
&&r(\nu _1, \nu _2|\nu _1', \nu _2')=p_1\;\;{\rm if}\;\;\nu _1= \nu_1' \;\;{\rm and}\;\;\nu _2\ne \nu_2'\nonumber\\
&&r(\nu _1, \nu _2|\nu _1', \nu _2')=p_2\;\;{\rm if}\;\;\nu _1\ne \nu_1' \;\;{\rm and}\;\;\nu _2= \nu_2'
\end{eqnarray}
In the first case both `unprimed factor bases' are different from the corresponding `primed factor bases' and therefore the result is always $1/(p_1p_2)$ (it cannot be zero).
In the second case the first `unprimed factor basis' is the same as the `primed factor basis', and therefore the result is $1/p_2$ or zero.
Analogous comment can be made for the last case.

\subsection{Factorization of the maximal lines in ${\cal G}(d)$ }

We represent a point $(\rho ,\sigma )$ in ${\mathbb Z}(d)\times {\mathbb Z}(d)$ as
\begin{eqnarray}
(\rho, \sigma)=(\overline \rho_1, \sigma_1)\times(\overline  \rho _2, \sigma _2);\;\;\;\;\overline \rho _i,  \sigma _i \in {\mathbb Z}(p_i)
\end{eqnarray}
Here we used the map of Eq.(\ref{map2}) for the first variable and the the map of Eq.(\ref{map1}) for the second variable.
The use of the two maps is important for the duality between maximal lines through the origin in ${\mathcal G}(d)$, and weak mutually unbiased bases in $H(d)$.
A maximal line $L(\rho ,\sigma )$ in ${\mathbb Z}(d)\times {\mathbb Z}(d)$ can now be factorized as
\begin{eqnarray}\label{line1}
L(\rho, \sigma)=L(\overline \rho_1,  \sigma_1)\times L(\overline  \rho _2, \sigma _2);\;\;\;\;\overline \rho _i,  \sigma _i \in {\mathbb Z}(p_i)
\end{eqnarray}
This is made clear in the following proposition.
\begin{proposition}
\mbox{}
\begin{itemize}
\item[(1)]
if $\overline \rho _i \ne 0\;({\rm mod}\;p_i)$, the $(\overline \rho _i)^{-1}$ exists in ${\mathbb Z}(p_i)$, and the line
$L(\overline\rho_i,  \sigma_i)=L(1, (\overline \rho _i) ^{-1}\sigma_i)$.
Also $L(\rho,  \sigma )=L(1, \rho  ^{-1}\sigma )$ and Eq.(\ref{line1}) can be written as
\begin{eqnarray}\label{line2}
&&L(1,\mu ^{-1}\nu)=L(1,\nu_1)\times L(1,\nu _2)\equiv {\cal L}(\nu _1,\nu _2)\nonumber\\
&&\nu=\nu _1s_1+\nu _2 s_2;\;\;\;\;\;\;\nu _i=(\overline\rho _i) ^{-1}\sigma_i\in {\mathbb Z}(p_i);\;\;\;\;\;\mu ^{-1}\nu =\rho ^{-1}\sigma\in {\mathbb Z}(p)\nonumber\\
&&\mu=p_1+p_2
\end{eqnarray}

\item[(2)]
If $\overline \rho _1=p_1=0 \;({\rm mod}\;p_1)$ then $\nu _1=-1$ by definition and 
\begin{eqnarray}\label{line3}
&&L(p_1,s_1+s_2\nu _2)=L(0,1)\times L(1,\nu _2)\equiv {\cal L} (-1,\nu _2);\;\;\;\;\;
\nu _2=(\overline \rho _2) ^{-1} \sigma_2
\end{eqnarray}
Similar result holds in the case that $\rho _2=p_2=0 \;({\rm mod}\;p_2)$:
\begin{eqnarray}\label{line30}
&&L(p_2,s_2+s_1\nu _1)=L(1,\nu _1)\times L(0,1)\equiv {\cal L} (\nu _1, -1);\;\;\;\;\;
\nu _1=(\overline\rho _1) ^{-1} \sigma_1
\end{eqnarray}

\item[(3)]
If $\rho _1=0 \;({\rm mod}\;p_1)$ and $\rho _2=0 \;({\rm mod}\;p_2)$ then $\nu _1=\nu _2=-1$ by definition and 
\begin{eqnarray}\label{line4}
&&L(0,1)=L(0,1)\times L(0,1)\equiv {\cal L} (-1,-1).
\end{eqnarray}
\end{itemize}
\end{proposition}

\begin{proof}
In all cases we show that the sets of points in the two sides are identical.
\begin{itemize}

\item[(1)]

\begin{eqnarray}
L(1,\nu_1)\times L(1,\nu_2)={\cal S}(0,-1|1,\nu _1)L(0,1)\times {\cal S}(0,-1|1,\nu _2)L(0,1)
\end{eqnarray}
Using Eq.(\ref{135}), and the fact that $L(0,1)\times L(0,1)$ is the line $L(0,1)$ in ${\cal G}(d)$, we get
\begin{eqnarray}
S(0,-\mu ^{-1}|\mu,\nu)L(0,1)=L(1,\mu ^{-1}\nu)
\end{eqnarray}
with the parameters given in Eq.(\ref{135}).
We used here Eq.(\ref{efb}).

\item[(2)]
\begin{eqnarray}
&&L(0,1)\times S(0,-1|1,\nu _2)L(0,1)=S(\kappa, \lambda|\mu,\nu)L(0,1)=
L(p_1, s_1+\nu _2s_2 )\nonumber\\
&&\kappa=s_1;\;\;\;\;
\lambda=-s_2p_1^{-1};\;\;\;\;
\mu=p_1;\;\;\;\;
\nu=s_1+\nu_2s_2
\end{eqnarray}
We used here Eq.(\ref{dfg}).
Comments analogous to remark \ref{zxd} are also valid for the lines.
\item[(3)]
This is straightforward.
\end{itemize}
\end{proof}
Therefore in ${\cal L}(\nu _1,\nu _2)$ the $\nu _i=-1,0,...,p_i-1$.
There are $\psi (d)=(p_1+1)(p_2+1)$
such lines through the origin , where $\psi(d)$ is the Dedekind $\psi$-function.
An example of this factorization for ${\cal G}(21)$ ($p_1=3$ and $p_2=7$) is given in table \ref{t1}.

Ref. \cite{W2} has shown that there exists a bijective map (duality) between the lines in ${\cal G}(d)$ and the weak mutually unbiased bases 
in $H(d)$ as follows:
\begin{eqnarray}\label{aaa}
{\cal B}(\nu _1, \nu _2)\;\leftrightarrow\; {\cal L} (\nu _1, \nu _2).
\end{eqnarray}
The finite geometry is non-near-linear geometry. The common points between two lines are described in the following proposition:
\begin{proposition}\label{1234}
Two maximal lines $L(1,\mu ^{-1}\nu)={\cal L}(\nu _1,\nu _2)$ and $L(1,\mu ^{-1}\nu')={\cal L}(\nu _1',\nu _2')$ (where $\nu=\nu_1s_1+\nu _2 s_2$ and $\nu'=\nu_1's_1+\nu _2' s_2$
through the origin, have in common $r(\nu _1,\nu _2|\nu _1',\nu _2')$ points where 
where $r(\nu _1,\nu _2|\nu _1',\nu _2')$ has been given in Eq.(\ref{200}).
\end{proposition}
\begin{proof}
The common points in the two lines should satisfy the relation
\begin{eqnarray}
(\lambda, \lambda \mu ^{-1}\nu)=(\lambda, \lambda \mu ^{-1}\nu')\;\;\rightarrow\;\;\lambda [(\nu_1-nu')s_1+(\nu _2-\nu _2') s_2]=0.
\end{eqnarray}
We first assume that $\nu _1\ne \nu _1'$ and $\nu _2\ne \nu _2'$. In this case $(\nu_1-\nu _1')s_1+(\nu _2-\nu _2') s_2$ is always different from zero, because the map
of Eq.(\ref{map1}) is bijective (and $0\leftrightarrow (0,0)$). Therefore in this case $\lambda =0$.

We next consider the case $\nu _1= \nu _1'$ and $\nu _2\ne \nu _2'$.
In this case any $\lambda$ which is multiple of $p_2$ gives a solution, because $p_2s_2=0$ (Eq.(\ref{20A})).
Therefore there are $p_1$ values of $\lambda$ which lead to common points.

The case $\nu _1 \ne \nu _1'$ and $\nu _2= \nu _2'$ is analogous to the above .
\end{proof}

Table \ref{t1} shows explicitly this duality for the case $d=21$.
In the present paper we show that there is another bijective map between these two sets and the set of zeros, in an analytic representation approach 
to weak mutually unbiased bases.

\begin{example}\label{exa}
We give an example of two lines through the origin in ${\cal G}(21)$, which have three points in common.
The lines $L(1,8)={\cal L} (2,3)$ and $L(1,11)={\cal L} (2,5)$ have in common the three points
$(0,0)$, $(7,14)$, $(14,7)$, and they are shown in Fig.\ref{f1}. This example shows that our geometry is a non-near-linear geometry.

The analogue of this in terms of bases is the ${\cal B} (2,3)$ and ${\cal B} (2,5)$.
In this case
\begin{eqnarray}
({\cal B}(2,3), {\cal B}(2,5))=\frac{3}{21}\;\;{\rm or}\;\;0.
\end{eqnarray}

Analogous example for two lines of zeros in ${\mathfrak Z}(21)$ is given later.

\end{example}

\section{Analytic representation of the weak mutually unbiased bases}

We first present a lemma which is needed in the proof of the proposition below.
\begin{lemma}
\begin{alignat}{1}\label{27y}
\prod _i\omega [\phi(\overline m_i,\overline k_i, \nu _i)]=\omega [\mu ^{-1}\phi(m,k, \nu)];\;\;\;\;\;\mu ^{-1}=p_2^{-1}s_1+p_1^{-1}s_2\;({\rm mod} \;d).
\end{alignat}
where $\phi(m,j,n)=-jm+2^{-1}\nu j^2$ (see Eq.(\ref{3cv})).
\end{lemma}
\begin{proof}
We use Eqs.(\ref{20A}) to prove that
\begin{equation}
km\mu ^{-1}=\overline m_1 \overline k_1 p_2+\overline m_2 \overline k_2 p_1;\;\;\;\;\;\nu_1({\overline k}_1)^2d_2+\nu_2({\overline k}_2)^2d_1=\mu ^{-1}\nu k^2.
\end{equation}
From these relations follows Eq.(\ref{27y}).
\end{proof}

We have explained earlier that Theta functions are Gaussian functions wrapped on a circle.
Symplectic transformations on Gaussian functions in a real line, give Gaussian functions. 
The proposition below proves an analogous statement for Theta functions.
This is needed later in order to prove that the zeros of the analytic representation of the state $\ket{{\cal X}(\nu _1,\nu _2);\overline m _1,\overline m _2}$
are on a straight line.
\begin{proposition}\label{1v9}
The analytic representation (defined in Eq.(\ref{aaa1})) of the state $\ket{{\cal X}(\nu _1,\nu _2);\overline m _1,\overline m _2}$
where $\nu _i=-1,...,p_i-1$ and $\overline m _i\in {\mathbb Z}(p_i)$, is given by:
\begin{itemize}
\item[(1)]
in the case $\nu _i=0,...,p_i-1$
\begin{eqnarray}\label{cvb}
&&|{\cal X}(\nu _1,\nu _2);\overline m_1,\overline m_2\rangle \quad \rightarrow \quad G(z)=\pi ^{-1/4}\exp \left (-\frac{ \pi}{d}z^2\right)
\Theta_3 (u; \tau)\nonumber\\
&&\tau=\frac{i-\nu\mu ^{-1}(d+1)}{d};\;\;\;\;\;u=-\pi \mu ^{-1}\left (\frac{\overline m_1}{p_1}+\frac{\overline m_2}{p_2}\right )+i\frac{\pi z}{p_1p_2}\nonumber\\
&&\nu=\nu _1s_1+\nu _2 s_2
\end{eqnarray}
where $\mu ^{-1},s_i$ are constants given in Eqs.(\ref{135}),(\ref{20}).
\item[(2)]
in the case $\nu _1=-1$ and $\nu _2=0,...,p_2-1$
\begin{eqnarray}\label{ggg}
 &&|{\cal X}(-1,\nu _2);\overline m_1,\overline m_2\rangle \quad \rightarrow \quad G(z)=\pi^{-1/4}\exp\left( \frac{-\pi p_2w^2}{p_1}\right) \Theta_3 (u;\tau)\nonumber\\
&&\tau=\frac{-\nu _2(p_2+1)+ip_1}{p_2};\;\;\;\;u=-\frac{\pi \overline m_2}{p_2}+i\pi w
;\;\;\;\;w=\frac{ z}{p_2}- \overline m_1.
\end{eqnarray}

Analogous result holds in the case $\nu _1=0,...,p_1-1$ and $\nu _2=-1$
\item[(3)]
in the case $\nu _1=\nu _2 =-1$
\begin{eqnarray}\label{cvb10}
|{\cal X}(-1,-1);m\rangle=|X;m\rangle \quad \rightarrow \quad G(z)=\pi ^{-1/4} \Theta_3\left[ \frac{\pi m}{ d}-z\left (\frac {\pi}{d}\right ) ; \frac{i}{d}\right]
\end{eqnarray}

\end{itemize}
\end{proposition}
\begin{proof}
\mbox{}
\begin{itemize}
\item[(1)]
Using Eq.(\ref{3cv}) with $\mu=1$ we get
 \begin{eqnarray}
 |X( \nu_i);\overline{m}_1\rangle = \frac{1}{ \sqrt{p_i}} \sum_{k_i=0}^{p_i-1} \omega[\phi ( \overline{m}_i,\overline{k}_i,\nu _i)]|X;\overline{k}_i\rangle
 \end{eqnarray}
\ \\
Therefore
\begin{alignat}{1}
|{\cal X}(\nu _1,\nu _2);\overline m_1,\overline m_2\rangle &=\sum_j|X;j\rangle\langle X;j|{\cal X}(\nu _1,\nu _2);\overline m_1,\overline m_2\rangle\nonumber \\
&=\sum_j|X;j\rangle[ \langle X_1;\overline j_1|\otimes \langle X_2;\overline j_2|]| {\cal X}(\nu _1, \nu _2);\overline m_1,\overline m_2\rangle\nonumber\\
&=\sum_j|X;j\rangle\left[\langle X_1;\overline j_1|X_1(\nu _1);\overline m_1\rangle\right] \left[\langle X_2;\overline j_2|X_2(\nu _2);\overline m_2\rangle\right] \nonumber \\
&=\sum_j\frac{1}{\sqrt d}\prod _i\omega [\phi(\overline m_i,\overline j_i, \nu _i)]|X;j\rangle .
\end{alignat}
We then use Eq.(\ref{27y}) and we get
\begin{eqnarray}
|{\cal X}(\nu _1,\nu _2);\overline m_1,\overline m_2\rangle =\sum_j\frac{1}{\sqrt d}\omega [\mu ^{-1}\phi(m,j, \nu)]|X;j\rangle .
\end{eqnarray}
Using lemma \ref{l} and Eq.(\ref{aaa1}), we represent $|X(\nu);m\rangle$ with the sum
\begin{eqnarray}\label{cvb34}
|{\cal X}(\nu _1,\nu _2);\overline m_1,\overline m_2\rangle \quad \rightarrow \quad \frac{ \pi^{-1/4}}{ \sqrt{d}}\sum_{j} \omega [ -\mu ^{-1} \phi (m,j,\nu)] \Theta_3
  \left[ \frac{\pi j}{ d}-z\left (\frac {\pi}{d}\right ); \frac{i}{d}\right]
\end{eqnarray}
We next show that this sum is equal to the Theta function shown on the right hand side of Eq.(\ref{cvb}).
Using the property of Theta functions in Eq.(\ref{y}) and Eq.(\ref{theta}), we find that
\begin{eqnarray}
&&\Theta_3 \left[ \frac{ \pi j}{ d}-z\left (\frac { \pi}{d}\right ) ; \frac{i}{d}\right]=\sqrt{d}\exp\left[ \frac{-\pi j^2}{d} + 2 j\left( \frac{ \pi}{d}\right)z -\left( \frac{ \pi}{d}z^2\right)\right]
\Theta_3\left[-i\pi j + i\pi z; id\right]\nonumber\\&&=\sqrt{d}\exp\left[ \frac{-\pi j^2}{d} + 2 j\left( \frac{ \pi}{d}\right)z -\left( \frac{ \pi}{d}z^2\right)\right]
\sum_{n= -\infty}^{ \infty} \exp\left[-\pi d n^2 + 2n\pi j - 2n\pi z\right].
\end{eqnarray}
In this paper we consider the case of odd $d$ and then $2^{-1}=\frac{d+1}{2}$. Therefore we get
\begin{eqnarray}
&&\frac{ \pi^{-1/4}}{ \sqrt{d}}\sum_{j}  \omega [-\mu ^{-1}\phi (m,j,\nu)]  \Theta_3
  \left[ \frac{\pi j}{ d}-z\left (\frac {\pi}{d}\right ) ; \frac{i}{d}\right]\nonumber\\
 &&= \pi^{-1/4}\exp\left[ \frac{-\pi}{d}z^2\right] \sum_{n= -\infty}^{ \infty} \sum_{j=0}^{d-1} \exp\left[\frac{-\pi}{d}(-j +nd)^2\right] \nonumber\\
 && \exp\left[-\frac{i\pi\nu\mu ^{-1} (d+1)}{d}(-j +nd)^2\right] \exp\left[-\frac{2i\pi m\mu ^{-1}}{d}(-j +nd)\right] \nonumber\\
 && \exp\left[-2(-j+nd)\left( \frac{ \pi}{d}\right)z\right]
\end{eqnarray}
We now change variable into $N=nd-j$. Since $n$ takes all integer values and $0\le j\le d-1$, the variable $N$ takes all integer values.
Therefore the above sum becomes
\begin{eqnarray}
 &&\pi^{-1/4}\exp\left[ \frac{-\pi}{d}z^2\right] \sum_{N= -\infty}^{ \infty} \exp\left[ \frac{-\pi }{d}N^2  -\frac{i\pi\nu \mu ^{-1}(d+1)}{d}N^2 - \frac{2i\pi m\mu ^{-1} N}{d} - 2N\left( \frac{ \pi}{d}\right)z\right].
\end{eqnarray}
This is the result in Eq.(\ref{cvb}).

\item[(2)]
We first point out that
\begin{eqnarray}
|{\cal X}(-1,\nu_2);\overline{ m}_1,\overline{m}_2\rangle =\sum_{k} \delta(\overline{k}_1,\overline{m}_1) \omega( \phi( \overline{ m}_2,\overline{k}_2,\nu_2) ) |X;k\rangle.
\end{eqnarray}
where $k=\overline k_1 p_2+\overline k_2 p_1$. Summation over $k$ is equivalent to summation over both $\overline k_1, \overline k_2$.

Its analytic representation is
\begin{eqnarray}
&&|{\cal X}(-1,\nu _2);\overline m_1,\overline m_2\rangle \quad \rightarrow \quad \frac{ \pi^{-1/4}}{ \sqrt{d}}\sum_{\overline k _1} \sum_{\overline k _2}\delta(\overline{k}_1,\overline{m}_1)\omega( -\phi( \overline{m}_2,\overline{k}_2,\nu_2) ) \Theta_3 \left[ \frac{\pi k}{ d}-z\left (\frac {\pi}{d}\right ); \frac{i}{d}\right]\nonumber\\&&=
\frac{ \pi^{-1/4}}{ \sqrt{d}}\sum_{\overline k _2}\omega( -\phi( \overline{m}_2,\overline{k}_2,\nu_2) )
\Theta_3 \left[ \frac{\pi (\overline m_1p_2+\overline k_2 p_1)}{ d}-z\left (\frac {\pi}{d}\right ); \frac{i}{d}\right]\nonumber\\
&&=\frac{ \pi^{-1/4}}{ \sqrt{d}}\sum_{\overline k _2}\omega( -\phi( \overline{m}_2,\overline{k}_2,\nu_2) )
\sqrt{d}\exp\left[ \frac{-\pi (\overline{m}_1p_2 + \overline{k}_2p_1)^2}{d} + 2 (\overline{m}_1p_2 + \overline{k}_2p_1)\left( \frac{ \pi}{d}\right)z -\left( \frac{ \pi}{d}z^2\right)\right]\nonumber\\&&
\times \Theta_3\left[-i\pi (\overline{m}_1p_2 + \overline{k}_2p_1) + i\pi z; id\right]\nonumber\\
 &&= \pi^{-1/4}\exp\left[ \frac{-\pi(\overline m_1p_2-z)^2}{p_1p_2}\right] \sum_{n= -\infty}^{ \infty} \sum_{ \overline{k}_2}
 \exp\left[-\frac{i\pi\nu _2 (p_2+1)}{p_2}(np_2- \overline{k}_2 )^2-\frac{2i\pi\overline{m}_2}{p_2}(np_2 - \overline{k}_2 )\right ] \nonumber\\
 && \times \exp\left[-\frac{\pi p_1}{p_2}(np_2-\overline k_2)^2+2\pi\left(\overline m_1-\frac{z}{p_2}\right)(np_2-\overline k_2)\right]
\end{eqnarray}

We now change variable into $N=np_2- \overline{k}_2 $, and we get the result in Eq.(\ref{ggg}).

\item[(3)]
Eq.(\ref{cvb10}) is obvious from the definition of the analytic representation.
\end{itemize}
\end{proof}
\begin{remark}
The $\tau$ in Eq.(\ref{cvb}) contains $\nu \mu ^{-1}$ which is an integer modulo $d$.
Consequently, $\tau$ is defined up to an integer multiple of $d+1$.
Since $d+1$ is an even integer, the $\Theta _3$ does not change (Eq.(\ref{theta})).
\end{remark}

Below we consider the states in WMUB $|{{\cal X}(\nu _1,\nu _2);\overline m_1,\overline m_2} \rangle$, and 
using proposition \ref{1v9}, we show that the zeros of their analytic representation are 
on a straight line.
\begin{proposition}\label{rrr}
The $d$ zeros of the analytic representation of the vector
$|{{\cal X}(\nu _1,\nu _2);\overline m_1,\overline m_2} \rangle$
where $\nu _i=-1,...,p_i-1$ and $\overline m_i\in {\mathbb Z}(p_i)$, are on a straight line and they are given by:
\begin{itemize}
\item[(1)]
in the case $\nu_i =0,...,p_i-1$ for $i=1,2$
\begin{eqnarray}\label{MM1}
&&\zeta (\nu_1,\nu _2 ;\overline m_1,\overline m_2; N)=\alpha(1-i\beta)+\gamma\nonumber\\
&&\alpha=N-\frac{1}{2};\;\;\;\;\beta=-\mu^{-1} \nu(d+1);\;\;\;\;\;\gamma=-idM+i\frac{d}{2}-i\mu ^{-1}m\nonumber\\
&&N=K+1,...,K+d;\;\;\;\;\;\;m={\overline m}_1p_2+{\overline m}_2p_1;\;\;\;\;\;\nu=\nu_1s_1+\nu_2 s_2
\end{eqnarray}
where $\mu ^{-1}, s_i$ are constants given in Eqs.(\ref{135}), (\ref{20}). 
Appropriate choices of the `winding integers' $K,M$, locate the zeros in the desirable cell.

\item[(2)]
in the case $\nu _1=-1$ and $\nu _2=0,...,p_2-1$
\begin{eqnarray}\label{620}
&&\zeta (-1,\nu _2; \overline m_1,\overline m_2; N)=\alpha (p_1-i\beta ')+\gamma '\nonumber\\
&&\alpha=N-\frac{1}{2};\;\;\;\;\beta '=-\nu _2(1+p_2);\;\;\;\;\;\gamma'={\overline m}_1p_2-
i{\overline m}_2-ip_2\left( M -\frac{1}{2}\right)\nonumber\\
&&N=K_1+1,...,K_1+p_2;\;\;\;\;\;M=K_2+1,...,K_2+p_1;\;\;\;\;\;\;m={\overline m}_1p_2+{\overline m}_2p_1.
\end{eqnarray}
Appropriate choices of the `winding integers' $K_1, K_2$, locate the zeros in the desirable cell.
Similar result holds for the case $\nu _2=-1$ and $\nu _1=0,...,p_1-1$.
\item[(3)]
in the case $\nu _1=\nu _2=-1$
\begin{eqnarray}\label{619}
&&\zeta (-1, -1;\overline m_1,\overline m_2; N)=-i\alpha +\gamma '';\;\;\;\;
\gamma ''=m-Md +\frac{d}{2}+id;\;\;\;\;\;N=K+1,...,K+d\nonumber\\
&&\alpha=N-\frac{1}{2};\;\;\;\;\;\;m={\overline m}_1p_2+{\overline m}_2p_1
\end{eqnarray}
Appropriate choices of the `winding integers' $K,M$, locate the zeros in the desirable cell.

\end{itemize}

\end{proposition}

\begin{proof}
\begin{itemize}
\item[(1)]
From Eq.(\ref{cvb}) we see that $|{\cal X}(\nu_1,\nu_2);\overline m_1,\overline m_2\rangle$ is represented by a single Theta function. Therefore the zeros in the case $\nu =\nu_1s_1+\nu_2s_2 =0,...,d-1$ are:
\begin{eqnarray}\label{eq3}
 -\frac{\pi \mu^{-1}m}{ d}+i\zeta \left (\frac {\pi}{d}\right )=(2M-1)\frac{\pi}{2}+(2N-1)\frac{\pi \tau}{2}
\end{eqnarray}
where $M,N$ are integers, and $\tau$ is given in Eq.(\ref{cvb}).
From this we get the result of Eqs.(\ref{MM1}).
\item[(2)]
In Eq.(\ref{ggg}) $|{\cal X}(-1,\nu_2);\overline m_1,\overline m_2\rangle$ is represented by a single Theta function. Therefore the zeros in the case $\nu_2 =0,...,p_2-1$ are:
\begin{eqnarray}\label{eq2}
 -\frac{\pi \overline {m}_2}{ p_2}+i\zeta \left (\frac {\pi}{p_2}\right )- i\pi\overline{m}_1=(2M-1)\frac{\pi}{2}+(2N-1)\frac{\pi \tau}{2}
\end{eqnarray}
where $M,N$ are integers, and $\tau$ is given in Eq.(\ref{ggg}).
From this we get the result of Eqs.(\ref{620}).
\item[(3)]
in the case $\nu_1=\nu_2 =-1$, the zeros of the Theta function in Eq.(\ref{cvb10}) give
\begin{eqnarray}\label{eq1}
 \frac{\pi m}{ d}-\zeta \left (\frac {\pi}{d}\right )=(2M-1)\frac{\pi}{2}+(2N-1)\frac{i\pi}{2d}
\end{eqnarray}
and from this follows Eq.(\ref{619}).
\end{itemize}
\end{proof}
In $\zeta (\nu_1,\nu _2;\overline m_1,\overline m_2; N)$ we used the `factorized notation' for the zeros corresponding to the vector $|{\cal X}(\nu_1,\nu_2);\overline m_1,\overline m_2\rangle$.
The correspondence between the two notations is given in Eqs(\ref{np1}),(\ref{np2}),(\ref{np3}),(\ref{np4}) for the various vectors, and from this follows that in the zeros in the unfactorized notation are
\begin{eqnarray}
&&\zeta (\nu _1,\nu _2;\overline m_1,\overline m_2)= {\cal \zeta} '(1,\mu ^{-1}\nu ;m \mu ^{-1})\nonumber\\
&&\nu=\nu _1s_1+\nu _2s_2;\;\;\;\;\;;m={\overline m}_1p_2+{\overline m}_2p_1;\;\;\;\;\nu_i=0,...,p_i-1
\end{eqnarray}
and 
\begin{eqnarray}
&&\zeta(-1,\nu _2;\overline m_1,\overline m_2)=\zeta '(p_1, s_1+\nu _2s_2 ;m);\;\;\;\;\;m={\overline m}_1p_2+{\overline m}_2p_1\nonumber\\
&&\zeta(\nu _1,-1;\overline m_1,\overline m_2)=\zeta '(p_2, s_2+\nu _1s_1 ;m)\nonumber\\
&&\zeta(-1,-1;\overline m_1,\overline m_2)=\zeta '(0,1;m)
\end{eqnarray}
We refer to the following set of $d$ zeros
\begin{eqnarray}
{\cal Z}(\nu _1,\nu_2;\overline m_1,\overline m_2)=\{\zeta (\nu _1,\nu_2; \overline m_1,\overline m_2; N);\;N=1,...,d\};\;\;\;\nu_i\in {\mathbb Z}(p_i)
\end{eqnarray}
as the `line' of the $d$ zeros corresponding to $\ket{{\cal X}(\nu _1,\nu_2);\overline m_1,\overline m_2}$.
In the unfactorized notation this is
\begin{eqnarray}
&&{\cal Z}'(1,\mu ^{-1}\nu ;m )=\{{\cal \zeta} '(1,\mu ^{-1}\nu ; m; N);\;N=1,...,d\}\nonumber\\
&&{\cal Z}'(p_1, s_1+\nu _2s_2 ;m)=\{{\cal \zeta} '(p_1, s_1+\nu _2s_2 ; m; N);\;N=1,...,d\}\nonumber\\
&&{\cal Z}'(p_2, s_2+\nu _1s_1 ;m)=\{\zeta '(p_2, s_2+\nu _1s_1 ; m; N);\;N=1,...,d\}\nonumber\\
&&{\cal Z}'(0,1;m)=\{{\cal \zeta} '(0,1;m;N);\;N=1,...,d\}
\end{eqnarray}

\begin{proposition}
The $d^2$ zeros in the cell ${\mathfrak S}$, of all $d$ vectors in the basis ${\cal B}(\nu _1,\nu _2)$ are 
\begin{eqnarray}\label{62}
{\mathfrak z}(r,s)=(r+is)+\frac{1}{2}(1+i);\;\;\;\;\;r,s=0,...,d-1.
\end{eqnarray}
and they do not depend on $(\nu _1, \nu_2)$.
We denote as ${\mathfrak Z}(d)$ the lattice of these zeros.
\end{proposition}
\begin{proof}
We consider three cases:
\begin{itemize}
\item[(1)]
In the case $\nu_i =0,...,p_i-1$ for $i=1,2$ we use Eq. (\ref{MM1}).
$N$ takes all values $1,...,d$ in the real axis. For each $N$, the $i\mu ^{-1}m$ gives all required values $1,...,d$ in the imaginary axis.
We note that when $m$ takes all values in ${\mathbb Z}(d)$, the $\mu ^{-1}m$ also takes all values in ${\mathbb Z}(d)$, because $\mu ^{-1}$ is invertible.

\item[(2)]
In the case $\nu _1=-1$ and $\nu _2=0,...,p_2-1$, we use Eq.(\ref{620}).
$Np_1+{\overline m}_1p_2$ takes all values $1,...,d$ in the real axis.
Indeed, $Np_1$ gives the integer multiples of $p_1$ and ${\overline m}_1p_2$ gives the `in between' values.
It is important here that $p_2$ is an invertible element within ${\mathbb Z}(p_1)$.

For each $Np_1+{\overline m}_1p_2$, the $i(p_2 M+{\overline m}_2)$ gives all required values $1,...,d$ in the imaginary axis.
Indeed, $Mp_2$ gives the integer multiples of $p_2$ and ${\overline m}_2$ gives the `in between' values.

Similar result holds for the case $\nu _2=-1$ and $\nu _1=0,...,p_1-1$.
\item[(3)]
In the case $\nu _1=\nu _2=-1$ we use Eq.(\ref{619}).
The $m$ takes all values $1,...,d$ in the real axis.
For each $m$, the $N$ gives all required values $1,...,d$ in the imaginary axis.

\end{itemize}
The above arguments do not depend on the value of $(\nu _1, \nu_2)$.
\end{proof}

\section{Triality between lines in finite geometries, WMUBs, and the zeros of their analytic representations}

\begin{definition}
${\cal A}(\nu _1,\nu _2)$ is the set of the $d$ parallel lines of zeros in ${\mathfrak S}$, of the $d$ vectors in a weak mutually unbiased basis:
\begin{eqnarray}
{\cal A}(\nu _1,\nu _2)=\{{\cal Z}(m;\nu _1,\nu_2)|m\in {\mathbb Z}(d)\};\;\;\;\nu_i\in{\mathbb Z}(p_i)
\end{eqnarray}
In the `unfactorized notation' this is
\begin{eqnarray}
&&A(1,\mu ^{-1}\nu )=\{{\cal Z} '(1,\mu ^{-1}\nu ;m )|m\in {\mathbb Z}(d)\}\nonumber\\
&&A(p_1, s_1+\nu _2s_2 )=\{{\cal Z} '(p_1, s_1+\nu _2s_2 ;m)|m\in {\mathbb Z}(d)\}\nonumber\\
&&A(p_2, s_2+\nu _1s_1 )=\{{\cal Z} '(p_2, s_2+\nu _1s_1 ;m)|m\in {\mathbb Z}(d)\}\nonumber\\
&&A(0,1)=\{{\cal Z} '(0,1;m)|m\in {\mathbb Z}(d)\}.
\end{eqnarray}
\end{definition}
Each of these sets is characterized by the slope of the lines it contains.
In the proposition below, we use the slopes of these lines. We also define slopes of a line $L(\rho, \sigma)$ in ${\cal G}(d)$ as $\frac{\sigma}{\rho}$.
Two lines $L(\rho, \sigma)$ and $L(\rho ', \sigma ')$ have the same slope if
\begin{eqnarray}\label{sdf}
\rho \sigma'-\rho '\sigma=0\;({\rm mod}\;d).
\end{eqnarray}
\begin{theorem}
\mbox{}
\begin{itemize}
\item[(1)]
There is a triality between
\begin{itemize}
\item
the weak mutually unbiased bases in $H(d)$
\item
the non-near linear finite geometry ${\mathcal G}(d)$ associated with the phase space ${\mathbb Z}(d)\times {\mathbb Z}(d)$
\item
the lattice ${\mathfrak Z}(d)$ in the cell ${\mathfrak S}$, which we also regard as a non-near linear finite geometry ${\mathbb Z}(d)$
\end{itemize}
as follows:
\begin{eqnarray}\label{trio}
{\cal B}(\nu _1,\nu _2)\;\leftrightarrow\;{\cal L}(\nu _1,\nu _2)\;\leftrightarrow\;{\cal A}(\nu _1,\nu _2) 
\end{eqnarray}
\item[(2)]
In this triality
\begin{itemize}
\item
the overlap between vectors in the WMUBs is 
$({\cal B}(\nu_1,\nu_2),{\cal B}(\nu_1',\nu_2'))=r(\nu _1, \nu _2|\nu _1', \nu _2')/d$, where $r(\nu _1, \nu _2|\nu _1', \nu _2')$ has been given in Eq.(\ref{200}).
\item
two lines ${\cal L}(\nu_1,\nu_2)$ and ${\cal L}(\nu_1',\nu_2')$ have in common $r(\nu _1, \nu _2|\nu _1', \nu _2')$ points 
\item
for any $m$, the lines ${\cal Z}(m; \nu_1,\nu_2)$ in ${\cal A}(\nu_1,\nu_2)$ and ${\cal Z}(m; \nu_1',\nu_2')$ in ${\cal A}(\nu_1',\nu_2')$
have $r(\nu _1, \nu _2|\nu _1', \nu _2')$ points in common
\end{itemize}
\end{itemize}
\end{theorem}
\begin{proof}
\mbox{}
\begin{itemize}
\item[(1)]
We have explained earlier (Eq.(\ref{aaa})) that ${\cal B}(\nu _1,\nu _2)\;\leftrightarrow\;{\cal L}(\nu _1,\nu _2)$ and we now prove that
${\cal L}(\nu _1,\nu _2)\;\leftrightarrow\;{\cal A}(\nu _1,\nu _2)$.
The proof is based on showing that the corresponding slopes are equal.
We consider the following three cases:
\begin{itemize}
\item
In the case $\nu _i=0,...,p_i-1$, Eq.(\ref{MM1}) shows that the slope of ${\cal A}(\nu _1,\nu _2)$ in ${\mathfrak Z}(d)$ is $\mu ^{-1}\nu(d+1)$.
Eq.(\ref{line2}) shows that the slope of the line ${\cal L}(\nu _1,\nu_2)=L(1,\mu ^{-1}\nu)$ in ${\cal G}(d)$ is also $\mu ^{-1}\nu$.
The two slopes are equal (modulo $d$). 
\item
In the case $\nu _1=-1$ and $\nu _2=0,...,p_2-1$, Eq.(\ref{620}) shows that the slope 
of ${\cal A}(-1,\nu _2)$ in ${\mathfrak Z}(d)$ is $\frac{\nu _2(1+p_2)}{p_1}$.
Eq.(\ref{line3}) shows that the slope of the line ${\cal L}(-1,\nu_2)=L(p_1,s_1+s_2\nu_2)$ in ${\cal G}(d)$ is $\frac{s_1+s_2\nu _2}{p1}$.
These two slopes are equal according to Eq.(\ref{sdf}).
Analogous result holds for the case $\nu _2=-1$ and $\nu _1=0,...,p_1-1$.
\item
In the case $\nu_1=\nu _2=-1$, Eq.(\ref{619}) shows that the ${\cal A}(-1,-1)$ in ${\mathfrak Z}(d)$ is vertical.
The line ${\cal L}(-1,-1)=L(0,1)$ in in ${\cal G}(d)$ is also vertical
\end{itemize}
\item[(2)]
In Eq.(\ref{jjj}), we have explained that $({\cal B}(\nu_1,\nu_2),{\cal B}(\nu_1',\nu_2'))=r(\nu _1, \nu _2|\nu _1', \nu _2')/d$.
Also in proposition \ref{1234} we have shown that two lines ${\cal L}(\nu_1,\nu_2)$ and ${\cal L}(\nu_1',\nu_2')$ have in common $r(\nu _1, \nu _2|\nu _1', \nu _2')$ points.
Below we prove analogous result for the lines of zeros.

We consider the lines ${\cal L}(\nu _1, \nu_2)$ and ${\cal L}(\nu _1', \nu_2')$ and assume that they have $r$ points in common (where $r|d$).
We show that the lines ${\cal Z}(m; \nu_1,\nu_2)$ and ${\cal Z}(m; \nu_1',\nu_2')$ also have $r$ points in common, i.e., that
$\zeta (m,\nu_1,\nu _2,N)=\zeta (m,\nu_1',\nu _2',N')$ for $r$ pairs $(N,N')$.
We give explicit proof only for the case that all $\nu_i,\nu_i'=0,...,d-1$. The proof in the other cases is similar. 

In this case, using Eq.(\ref{line2}) we conclude that there exist $r$ pairs $(\lambda , \lambda ')$ such that
\begin{eqnarray}
(\lambda , \lambda \mu ^{-1}\nu )=(\lambda', \lambda ' \mu ^{-1}\nu ');\;\;\;\;\;\nu=\nu _1s_1+\nu _2s_2.
\end{eqnarray}
This leads to $\lambda =\lambda '\;({\rm mod}\;d)$ and  $\lambda \mu ^{-1}\nu =\lambda' \mu ^{-1}\nu ' \;({\rm mod}\;d)$.
We then use Eq.(\ref{MM1}) to prove that
\begin{eqnarray}
\zeta (m,\nu_1,\nu _2,\lambda )=\zeta (m,\nu_1',\nu _2',\lambda ').
\end{eqnarray}
for each of the $r$ pairs $(\lambda , \lambda ')$.

\end{itemize}
\end{proof}
Table \ref{t1} shows explicitly this triality for the case $d=21$.
The precise correspondence of the various quantities involved in this triality, is summarized in table \ref{t2}.

\begin{example}
We consider an example of two sets of lines of zeros in ${\mathfrak Z}(21)$, which is analogous to example \ref{exa} (in this case $p_1=3$ and $p_2=7$).
They are the ${\cal A} (2,3)$ and ${\cal A} (2,5)$.
We take the line of zeros ${\cal Z}(4;2,3)$ from the set ${\cal A} (2,3)$,
and the line of zeros ${\cal Z}(4;2,5)$ from the set ${\cal A} (2,5)$ (i.e., we take as an example, $m=4$).
The lines ${\cal Z}(4;2,3)$ and ${\cal Z}(4;2,5)$ are shown in Fig.\ref{f2}, and they
have in common the $p_1=3$ zeros:
\begin{eqnarray}
N=4\;\rightarrow\;{\zeta}(4;2,3,N)={\zeta}(4;2,5;N)=3.5+i4.5\nonumber\\
N=4+p_2=11\;\rightarrow\;{\zeta}(4;2,3,N)={\zeta}(4;2,5;N)=10.5+i18.5\nonumber\\
N=4+2p_2=18\;\rightarrow\;{\zeta}(4;2,3,N)={\zeta}(4;2,5;N)=17.5+i11.5
\end{eqnarray}
If we regard the $3.5+i4.5$ as `origin', these three points have coordinates $(0,0)$, $(7,14)$ and $(14,7)$,
which are exactly the same as in the example \ref{exa}.
Comparison of Figs.\ref{f1},\ref{f2} shows this.

We also consider the case $m=5$.
The lines ${\cal Z}(5;2,3)$ and ${\cal Z}(5;2,5)$ have in common the $p_1=3$ zeros:
\begin{eqnarray}
N=4\;\rightarrow\;{\zeta}(5;2,3,N)={\zeta}(5;2,5;N)=3.5+i6.5\nonumber\\
N=4+p_2=11\;\rightarrow\;{\zeta}(5;2,3,N)={\zeta}(5;2,5;N)=10.5+i20.5\nonumber\\
N=4+2p_2=18\;\rightarrow\;{\zeta}(5;2,3,N)={\zeta}(5;2,5;N)=17.5+i13.5
\end{eqnarray}
Again we regard the $3.5+i6.5$ as `origin', and these three points have coordinates $(0,0)$, $(7,14)$ and $(14,7)$,
as above and as in the example \ref{exa}.

It is seen that for any $m$, the lines ${\cal Z}(m; \nu_1,\nu_2)$ in ${\cal A}(\nu_1,\nu_2)$ and ${\cal Z}(m; \nu_1',\nu_2')$ in ${\cal A}(\nu_1',\nu_2')$ have $r$ points in common (where $r=1,p_1,p_2$).
\end{example}

\section{Discussion}

The objective of this paper is to use analytic representations and their zeros, in the study of the general area of mutually unbiased bases.
Quantum states are represented with the analytic functions of Eq.(\ref{aaa1}).
The zeros of these analytic functions determine uniquely the quantum state of the system.

We have shown that there is a triality that links lines in ${\cal G}(d)$, WMUBs in $H(d)$, and lines of zeros of WMUBs in  ${\mathfrak Z}(d)$.
The duality between lines in ${\cal G}(d)$, and WMUBs in $H(d)$ is surprising, but with hindsight it might be argued that
quantum states in the Hilbert space inherit the properties of the underline phase space.
But the appearance of lines of zeros of WMUBs in  ${\mathfrak Z}(d)$ as a third component in this triality is certainly very surprising, and it reaffirms
the important (albeit counterintuitive) role of analytic functions in the description of quantum systems. 

The general methodology is the factorization of a $d$-dimensional system into subsystems with prime dimension
(for simplicity we have taken $d=p_1p_2$). 
Tensor products of mutually unbiased bases in each subsystem lead to weak mutually unbiased bases, with overlaps given in Eq.(\ref{8}).
There is a duality between weak mutually unbiased bases, and maximal lines through the origin in the  
${\cal G}(d)={\mathbb Z}(d) \times {\mathbb Z}(d)$ phase space.
This duality has been extended in this paper into a triality, with the involvement of the zeros of analytic functions that represent the quantum states.

The method can also be used when 
$d=p_1\times...\times p_N$. In this case,
there is an isomorphism between $H(d)$ and $H(p_1)\otimes...\otimes H(p_N)$, and the weak mutually unbiased bases are tensor products of
mutually unbiased bases in each $H(p_i)$. 
Bijective maps between ${\mathbb Z}(d)$ and ${\mathbb Z}(p_1)\times ...\times {\mathbb Z}(p_N)$ 
(generalizations of Eqs.(\ref{map1}), (\ref{map2})) can be found in \cite{2} (and in \cite{Good}).
Using them we can factorize the lines in the finite geometry ${\mathbb Z}(d)\times {\mathbb Z}(d)$.
We can also define analytic representations in a cell in the complex plane, and
factorize the lines of their zeros.
The methodology here is analogous to the one that we presented, but the technical details are more complicated.

Existing work in the general area of mutually unbiased bases is based on discrete Mathematics. The present work links them with the theory of analytic functions.

\newpage
\begin{table}
\centering
\caption{Correspondence between WMUBs in the Hilbert space $H(21)$, lines in the ${\cal G}(21)$ phase space,  and sets of lines of zeros in the lattice ${\mathfrak Z}(21)$.
Both the `factorized notation' and `unfactorized notation' are shown.}
\begin{tabular}{|p{4cm}|p{4cm}|p{4cm}|}
\hline
$H(21)$ & ${\cal G}(21)$ & $ \mathfrak {Z}(21)$ \\
\hline
{ \begin{align*} B(0,1) &= \mathcal{B}(-1,-1) \\ B(1,0) &= \mathcal{B}(0,0) \\ B(1,1) &= \mathcal{B}(1,3) \\ B(1,2)  &= \mathcal{B}(2,6)  \\  B(1,3) &= \mathcal{B}(0,2) \\ B(1,4) &= \mathcal{B}(1,5) \\ B(1,5)  &=  \mathcal{B}(2,1) \\ B(1,6) &= \mathcal{B}(0,4)   \\ B(1,7)&= \mathcal{B}(1,0) \\ B(1,8) &= \mathcal{B}(2,3)  \\ B(1,9) &=  \mathcal{B}(0,6) \\ B(1,10) &= \mathcal{B}(1,2)  \\ {B}(1,11) &= \mathcal{B}(2,5) \\ {B}(1,12) &= \mathcal{B}(0,1) \\ B(1,13) &= \mathcal{B}(1,4) \\ B(1,14) &= \mathcal{B}(2,0) \\ B(1,15) &= \mathcal{B}(0,3) \\
 B(1,16) &= \mathcal{B}(1,6) \\ B(1,17) &= \mathcal{B}(2,2) \\ B(1,18) &= \mathcal{B}(0,5) \\ B(1,19) &= \mathcal{B}(1,1) \\  B(1,20) &=  \mathcal{B}(2,4) \\ B(3,7) &= \mathcal{B}(-1,0) \\ B(3,1) &= \mathcal{B}(-1,1) \\  B(3,16) &= \mathcal{B}(-1,2) \\ B(3,10) &= \mathcal{B}(-1,3) \\ B(3,4) &= \mathcal{B}(-1,4) \\ B(3,19) &= \mathcal{B}(-1,5) \\ B(3,13)  &=  \mathcal{B}(-1,6) \\ B(7,15) &= \mathcal{B}(0,-1) \\ B(7,1)&= \mathcal{B}(1,-1) \\ B(7,8) &= \mathcal{B}(2,-1) \\
 \end{align*}}
&
{ \begin{align*} L(0,1) &= {\cal L}(-1,-1) \\ L(1,0) &= {\cal L}(0,0) \\ L(1,1) &= {\cal L}(1,3) \\ L(1,2)  &= {\cal L}(2,6)  \\  L(1,3) &= {\cal L}(0,2) \\ L(1,4) &= {\cal L}(1,5) \\ L(1,5)  &=  {\cal L}(2,1) \\ L(1,6) &= {\cal L}(0,4)   \\ L(1,7)&= {\cal L}(1,0) \\ L(1,8) &= {\cal L}(2,3)  \\ L(1,9) &=  {\cal L}(0,6) \\
 L(1,10) &= {\cal L}(1,2)  \\ L(1,11) &= {\cal L}(2,5) \\ L(1,12) &= {\cal L}(0,1) \\ L(1,13) &= {\cal L}(1,4) \\ L(1,14) &= {\cal L}(2,0) \\ L(1,15) &= {\cal L}(0,3) \\
 L(1,16) &= {\cal L}(1,6) \\ L(1,17) &= {\cal L}(2,2) \\ L(1,18) &= {\cal L}(0,5) \\ L(1,19) &={\cal L}(1,1) \\  L(1,20) &=  {\cal L}(2,4) \\
  L(3,7) &= {\cal L}(-1,0) \\ L(3,1) &= {\cal L}(-1,1) \\  L(3,16) &= {\cal L}(-1,2) \\L(3,10) &= {\cal L}(-1,3) \\ L(3,4) &= {\cal L}(-1,4) \\ L(3,19) &= {\cal L}(-1,5) \\ L(3,13)  &=  {\cal L}(-1,6) \\ L(7,15) &= {\cal L}(0,-1) \\ L(7,1)&= {\cal L}(1,-1) \\ L(7,8) &= {\cal L}(2,-1) \\
 \end{align*}}
&

{ \begin{align*} A(0,1) &= {\cal A}(-1,-1) \\ A(1,0) &= {\cal A}(0,0) \\ A(1,1) &= {\cal A}(1,3) \\ A(1,2)  &= {\cal A}(2,6)  \\  A(1,3) &= {\cal A}(0,2) \\ A(1,4) &= {\cal A}(1,5) \\ A(1,5)  &=  {\cal A}(2,1) \\ A(1,6) &= {\cal A}(0,4)   \\ A(1,7)&= {\cal A}(1,0) \\ A(1,8) &= {\cal A}(2,3)  \\ A(1,9) &=  {\cal A}(0,6) \\
 A(1,10) &= {\cal A}(1,2)  \\ A(1,11) &= {\cal A}(2,5) \\ A(1,12) &= {\cal A}(0,1) \\ A(1,13) &= {\cal A}(1,4) \\ A(1,14) &= {\cal A}(2,0) \\ A(1,15) &= {\cal A}(0,3) \\
 A(1,16) &= {\cal A}(1,6) \\ A(1,17) &= {\cal A}(2,2) \\ A(1,18) &= {\cal A}(0,5) \\ A(1,19) &={\cal A}(1,1) \\  A(1,20) &=  {\cal A}(2,4) \\
  A(3,7) &= {\cal A}(-1,0) \\ A(3,1) &= {\cal A}(-1,1) \\  A(3,16) &= {\cal A}(-1,2) \\A(3,10) &= {\cal A}(-1,3) \\ A(3,4) &= {\cal A}(-1,4) \\ A(3,19) &= {\cal A}(-1,5) \\ A(3,13)  &=  {\cal A}(-1,6) \\ A(7,15) &= {\cal A}(0,-1) \\ A(7,1)&= {\cal A}(1,-1) \\ A(7,8) &= {\cal A}(2,-1) \\
 \end{align*}}
 \\
\hline
\end{tabular}\label{t1}
\end{table}

\begin{table}[h]
\renewcommand{\arraystretch}{2}
\centering
\caption{Correspondence between the various quantities in the triality of Eq.(\ref{trio}).}
\begin{tabular}{|p{0.33\columnwidth}|p{0.33\columnwidth}|p{0.33\columnwidth}|}
\hline
{\bf WMUBs in $H(d)$}&{\bf Lines in ${\cal G}(d)$} &  {\bf Lines of zeros in ${\mathfrak Z}(d)$}\\
\hline
$\psi(d)$ WMUB ${\cal B}(\nu_1,\nu_2)$&$\psi(d)$ maximal lines through the origin ${\cal L}(\nu_1,\nu_2)$&
$\psi (d)$ sets ${\cal A}(\nu_1,\nu_2)$ of parallel lines of zeros
\\\hline
$d$ orthogonal vectors in each WMUB ${\cal B}(\nu_1,\nu_2)$&
$d$ points in each ${\cal L}(\nu_1,\nu_2)$&
$d$ parallel lines of zeros in each set ${\cal A}(\nu_1,\nu_2)$. Each line contains $d$ zeros.
\\\hline
{\small $({\cal B}(\nu_1,\nu_2),{\cal B}(\nu_1',\nu_2'))=\frac{r(\nu _1, \nu _2|\nu _1', \nu _2')}{d}$ }

(Eq.(\ref{200}))&
two lines ${\cal L}(\nu_1,\nu_2)$ and ${\cal L}(\nu_1',\nu_2')$ have in common $r(\nu _1, \nu _2|\nu _1', \nu _2')$ points& for any $m$, the lines ${\cal Z}(m; \nu_1,\nu_2)$ in ${\cal A}(\nu_1,\nu_2)$ and ${\cal Z}(m; \nu_1',\nu_2')$ in ${\cal A}(\nu_1',\nu_2')$
have $r(\nu _1, \nu _2|\nu _1', \nu _2')$ points in common\\\hline
\end{tabular}\label{t2}
\end{table}
\newpage

\begin{figure}
\caption{The lines ${\cal L}(2,3)$ (circles), and ${\cal L}(2,5)$ (crosses), in the
${\cal G} (21)$ finite geometry.
The two lines have in common the three points $(0,0)$, $(7,14)$, $(14,7)$.}
\includegraphics[scale = 1.0]{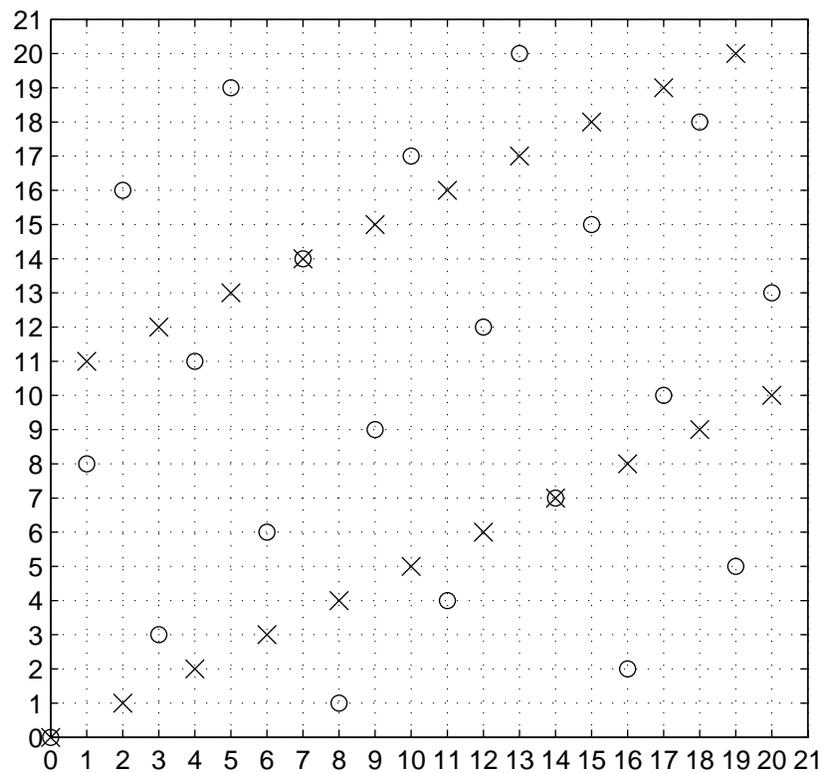}
\label{f1}
\end{figure}

\begin{figure}
\caption{The lines of zeros ${\cal Z}(4;2,3)$ (circles), and ${\cal Z}(4;2,5)$ (crosses), in the 
cell ${\mathfrak S}$ in the complex plane .
The two lines have in common the zeros $3.5+i4.5$, $10.5+i18.5$, $17.5+i11.5$.}
\includegraphics[scale = 1.0]{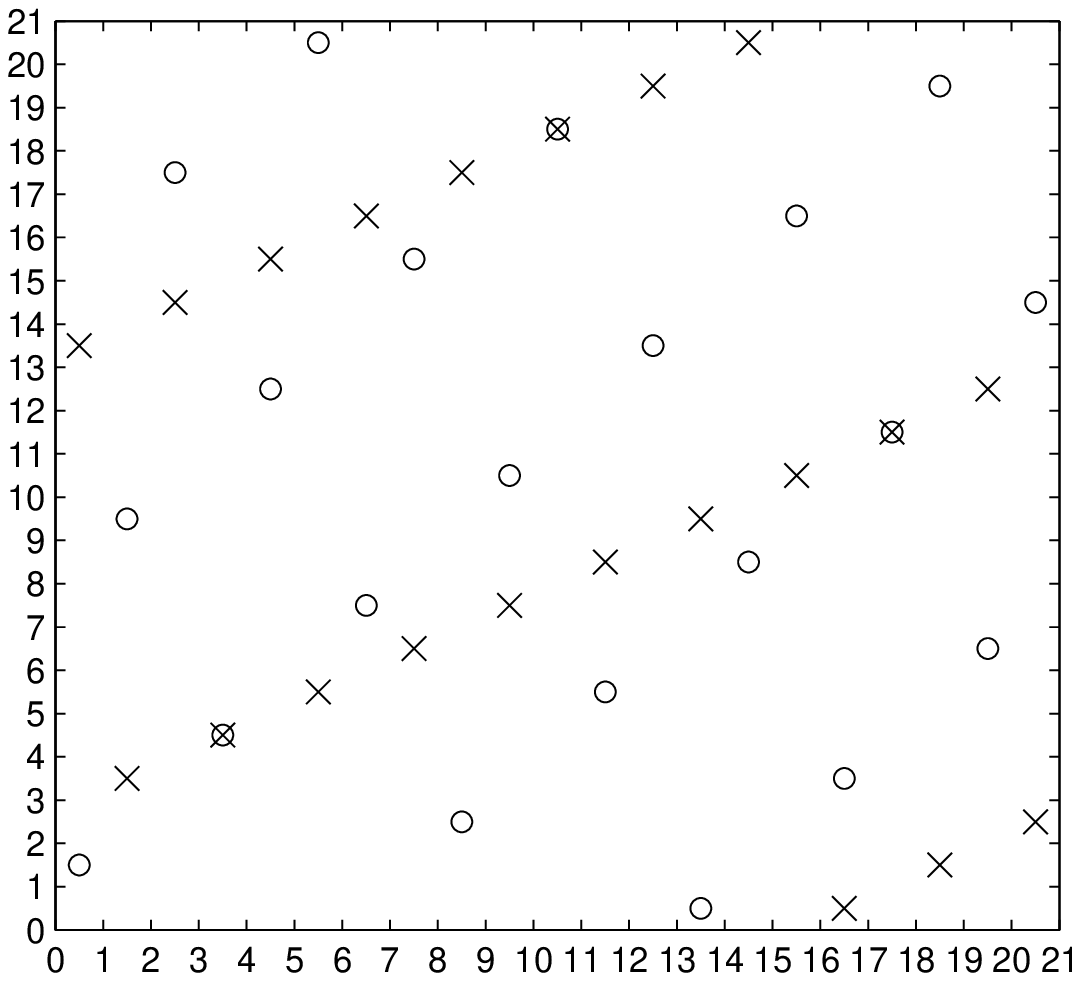}
\label{f2}
\end{figure}


\begin{thebibliography}{90}
\bibitem{SCH}
J. Schwinger, `Quantum Kinematics and Dynamics' (Benjamin, New York, 1970)
\bibitem{1}
P. Stovicek, J. Tolar,Rep. Math. Phys. 20, 157 (1984) 
\bibitem{2}
A. Vourdas, Rep. Prog. Phys. 67, 1 (2004)
\bibitem{3}
A. Vourdas, J. Phys. A40, R285 (2007)
\bibitem{4}
G. Bjork, A.B. Klimov, L.L. Sanchez-Soto, Prog. Optics 51, 469 (2008)
\bibitem{5}
M. Kibler, J. Phys. A42, 353001 (2009)
\bibitem{6}
N. Cotfas, J.P. Gazeau, J.Phys. A43, 193001(2010)
\bibitem{7}
T. Durt, B.G. Englert, I. Bengtsson, K. Zyczkowski, Int. J. Quantum Comp. 8, 535 (2010)
\bibitem{m1}
W. Wootters, B.D. Fields, Ann. Phys. (NY), 191, 363 (1989)
\bibitem{m2}   
S. Bandyopadhyay, P.O. Boykin, V.Roychowdhury, F. Vatan, Algorithmica 34, 512 (2002)
\bibitem{m3}
K. Gibbons, M.J. Hoffman, W. Wootters, Phys. Rev. A70, 062101 (2004)
\bibitem{m4}	
A. Klappenecker, M. Rotteler, Lect. Notes Comp. Science 2948, 137 (2004)
\bibitem{m5}
J.L. Romero, G. Bjork, A.B. Klimov, L.L. Sanchez-Soto, Phys. Rev. A72, 062310 (2005)
\bibitem{m6}	
M. Saniga, M. Planat, J. Phys. A39, 435 (2006)
\bibitem{m7}
P. Sulc, J. Tolar, J. Phys. A40, 15099 (2007)
\bibitem{m8}
J. Tolar, G. Chadzitaskos, J. Phys. A42, 245306 (2009)
\bibitem{m9}
S. Brierley, S. Weigert, I. Bengtsson, Quant. Inf. Comp., 10, 803 (2010)
\bibitem{m10}
W. Kantor, J. Math. Phys. 53, 032204 (2012)
\bibitem{B}
T. Beth, D. Jungnickel, H. Lenz, `Design Theory' (Cambridge Univ. Press, Cambridge, 1993)
\bibitem{B1}
G. Zauner, Int. J. Quantum Inf. 9, 445 (2011)
\bibitem{LS}
J. Denes, A.D. Keedwell, `Latin Squares and their Applications' (Academic, New York, 1974)
\bibitem{W1}
M. Shalaby, A. Vourdas, J. Phys. A45, 052001 (2012)
\bibitem{W2}
M. Shalaby, A. Vourdas, Ann. Phys. 337, 208 (2013) 
\bibitem{f1}
L.M. Batten, `Combinatorics of finite geometries', Cambridge Univ. Press, Cambridge, 1997
\bibitem{f2}
J.W.P. Hirchfeld, `Projective geomtries over finite fields' (Oxford Univ. Press, Oxford, 1979)
\bibitem{f3}
J.W.P. Hirchfeld, J.A. Thas, `General Galois geometries' (Oxford Univ. Press, Oxford, 1991)
\bibitem{r1}
M. Planat, M. Saniga, M. R. Kibler, SIGMA, 2, 66 (2006)
\bibitem{r2}
H. Havlicek, M. Saniga, J. Phys. A40, F943 (2007)
\bibitem{r3}
H. Havlicek, M. Saniga, J. Phys. A41, 015302 (2008)
\bibitem{r4}
M. Korbelar, J. Tolar, J. Phys. A43, 375302 (2010) 
\bibitem{A0}
V. Bargmann, Commun. Pure Appl. Math. {\bf 14}, 187 (1961)
\bibitem{A00}
V. Bargmann, Commun. Pure Appl. Math. {\bf 20}, 1 (1967)
\bibitem{A1}
A. Perelomov, `Generalized coherent states and their applications', (Springer, Berlin, 1986) 
\bibitem{A2}
B.C. Hall, Contemp. Math. 260, 1 (2000)
\bibitem{A3}
A. Vourdas, J. Phys. A39, R65 (2006) 
\bibitem{A4}
P. Leboeuf, J. Phys. A24, 4575 (1991)
\bibitem{A5}
M.B. Cibils, Y. Cuche, P. Leboeuf, W.F. Wreszinski, Phys. Rev A46, 4560 (1992)
\bibitem{A6}
S. Nonnenmacher, A. Voros, J. Phys. A30, 295 (1997)
\bibitem{A7}
H.J. Korsch, C. M\'uller, H. Wiescher, J. Phys. A30, L677 (1997)
\bibitem{A8}
A. Nonnenmacher, A. Voros, J. Stat. Phys. 92, 431 (1998)
\bibitem{A9}
F. Toscano, A.M.O. de Almeida, J. Phys. A32, 6321 (1999) 
\bibitem{A10}
D. Biswas, S. Sinha, Phys. Rev. E60, 408 (1999)
\bibitem{C1}
R.P. Boas "Entire functions" (Academic,New York,1954)
\bibitem{C2}
B.Ja. Levin,"Distribution of zeros of entire functions" (American Math. Soc, Rhode Island, 1964)
\bibitem{C3}
B.Ja. Levin,"Lectures on entire functions" (American Math.Soc, Rhode Island, 1996)
\bibitem{AN4}
S. Zhang, A. Vourdas, J. Phys. A37, 8349 (2004); and corrigendum in J. Phys. A38, 1197 (2005) 
\bibitem{AN5}
M. Tabuni, S. Zhang, A. Vourdas, Phys. Scr. 82, 038107 (2010)
\bibitem{THETA}
D Mumford, `Tata lectures on Theta', Vols 1,2,3 (Birkhauser, Boston, 1983)
\bibitem{AN6}
P. Leboeuf, A. Voros, J. Phys. A23, 1765 (1990)
\bibitem{zak}
J. Zak, J. Phys. A44, 345305 (2011) 
\bibitem{Good}
I.J. Good, IEEE Trans. Computers, C-20, 310 (1971)
\bibitem{SV2}
M. Shalaby, A. Vourdas, J. Phys. A44, 345303 (2011)

\end{thebibliography}
\end{document}